\documentclass[10pt]{amsart}

\usepackage{amsfonts}
\usepackage{amssymb}
\usepackage{amsthm}
\usepackage{color}
\usepackage{latexsym}
\usepackage{bbm}
\usepackage{enumerate}
\usepackage{graphicx}
\usepackage{graphics}
\usepackage{float}
\usepackage{sidecap}
\usepackage{wrapfig}
\usepackage{subfig}

\usepackage{amsmath}
\usepackage{setspace}
\usepackage[tmargin=2.0cm, bmargin=2.0cm, lmargin=2.0cm, rmargin=2.0cm]{geometry}

\newcommand*\diff{\mathop{}\!\mathrm{d}}

\newtheorem{theorem}{Theorem}

\newtheorem{lem}[theorem]{Lemma}

\newtheorem{Def}[theorem]{Definition}
\newtheorem{Prop}[theorem]{Proposition}

\theoremstyle{definition}
\newtheorem{Ex}{Example}
\theoremstyle{definition}
\newtheorem{Rem}[theorem]{Remark}

\def\E{{\mathbb E}} 

\def\WW{\mathcal{W}}
\def\ZZ{\mathcal{Z}}

\title{Fair valuation of L\'evy-type drawdown-drawup contracts with general insured and penalty functions}

\author{Zbigniew Palmowski}
\address{Faculty of Pure and Applied Mathematics, Wroc\l aw University of Science and Technology, ul. Wyb. Wyspia\'nskiego 27, 50-370 Wroc\l aw, Poland}
\email{zbigniew.palmowski@gmail.com}

\author{Joanna Tumilewicz}
\address{Mathematical Institute, University of Wroc\l aw, pl. Grunwaldzki 2/4, 50-384 Wroc\l aw, Poland}
\email{joanna.tumilewicz@gmail.com}

\thanks{This work is partially supported by National Science Centre Grants No. 2015/17/B/ST1/01102
(2016-2019) and No. 2016/23/N/ST1/01189
(2017-2019).}

\date{\today}

\begin{document}

\begin{abstract}
In this paper, we analyse some equity-linked contracts that are related to drawdown and drawup events based on assets governed by a geometric spectrally negative L\'evy process.
Drawdown and drawup refer to the differences between the historical maximum and minimum of the asset price and its current value, respectively.
We consider four contracts.
In the first contract, a protection buyer pays a premium with a constant intensity $p$ until the drawdown of fixed size occurs.
In return, he/she receives a certain insured amount at the drawdown epoch, which depends on the drawdown level at that moment.
Next, the insurance contract may expire earlier if a certain fixed drawup event occurs prior to the fixed drawdown.
The last two contracts are extensions of the previous ones but with an additional cancellable feature that
allows the investor to terminate the contracts earlier.
In these cases, a fee for early stopping depends on the drawdown level at the stopping epoch.
In this work, we focus on two problems: calculating the fair premium $p$ for basic contracts and finding the optimal stopping rule for the polices with a cancellable feature.
To do this, we use a fluctuation theory of L\'evy processes and rely on a theory of optimal stopping.
\vspace{3mm}


\end{abstract}

\maketitle

\pagestyle{myheadings} \markboth{\sc Z.\ Palmowski
--- J.\ Tumilewicz} {\sc Fair valuation of L\'evy-type drawdown-drawup contracts}

\vspace{1.8cm}

\tableofcontents

\newpage

\section{Introduction}

The recent financial crises have shown that drawdown events can directly affect the incomes of individual and institutional investors.
This basic observation suggests that drawdown protection can be very useful in daily practice.
For this reason, in this paper we consider  a few insurance contracts that can protect the buyer from a large drawdown.
By the drawdown of a price process, we mean here the distance of the current value  from the maximum value that it has attained to date.
In return for the protection, the investor pays a premium.
More precisely, we consider the following insurance contracts.
In the simplest contract, the protection buyer pays a premium with a constant intensity until the drawdown of fixed size occurs.
In return, he/she receives a certain insured amount at the drawdown epoch that depends on the level of drawdown at this moment.
Another insurance contract provides protection from any specified drawdown with drawup contingency.
The drawup is defined as the current rise of the asset present value over the running minimum.
This contract may expire earlier if a certain fixed drawup event occurs prior to the fixed drawdown.
This is a very demanding feature of an insurance contract from the investor's perspective.
Indeed, when a large drawup is realised, there is little need to insure against a drawdown.
Therefore, this drawup contingency automatically stops the
premium payment and is an attractive feature that will potentially
reduce the cost of drawdown insurance.

In fact, the buyer of the insurance contract might think that they are unlikely to get large drawdown and he/she
might want to stop paying the premium at some other random time.
Therefore, we expand the previous two contracts by adding a cancellable feature.
In this case,
the fee for early stopping depends on the level of drawdown at the stopping epoch.

We focus on two problems: calculating the fair premium $p^*$ for basic contracts and showing
that the investor's optimal cancellable timing is based on the first passage time of the drawdown process.
This allows us to identify the fair price of all of the contracts that we have described.

The shortcomings of the diffusion models in representing the risk related to large market
movements have led to the development of various option pricing models with jumps,
where large log-returns are represented as discontinuities in prices as a function of time.
Therefore, in this paper we model an asset price appearing in these  contracts
with a geometric spectrally negative L\'evy process.
In this model, the log-price $\log S_t=X_t$ is described by a L\'evy process without positive jumps.
This is a  natural generalisation of the Black-Scholes market (for which $X_t=B_t$ is a Brownian motion), which
allows for a more realistic representation
of price dynamics and a greater flexibility in calibrating the model to market prices.
This will also allow us to reproduce a wide variety of implied volatility skews and smiles
(see e.g. \cite{Cont}).

In this paper we follow Zhang et al. \cite{olympia} and Palmowski and Tumilewicz \cite{ZP&JT}.
Zhang et al. \cite{olympia} considered the Black-Scholes model, in contrast to our
more general, L\'evy-type market. However, they did not consider
an insurance contract with a drawup contingency and cancellable feature.

In Zhang et al. \cite{olympia}, and Palmowski and Tumilewicz \cite{ZP&JT} the insured amount and penalty fee are fixed and constant.
In this paper, we allow these quantities to depend on level of drawdown at the maturity of the contract or at the stopping epoch.
This new feature in our model allows for more flexible insurance contracts.
Analysing this interesting case also requires a deeper understanding of the position of the L\'evy process
at these stopping times, which is also of theoretical interest.
Apparently, this could be achieved by using the fluctuation theory of spectrally negative L\'evy processes,
and can refine and find new results from the optimal stopping theory.
The research conducted in this paper continues a list of papers analysing
drawdown and drawup processes, see for example \cite{CZH, GZ, MA, drawdownup2, PV, Sorn, Vec1,Vec2, drawdownup1}.

In this paper, we also give an extensive numerical analysis
which shows that suggested optimal stopping times and fair premium rule are easy to
find and the implemented algorithm is very efficient.
We mainly focus on the case where a logarithm of the asset price is a linear Brownian motion
(Black-Scholes model) or drift minus compound Poisson process (so-called Cram\'er-Lundberg risk process).
The dependency of the price of the considered contracts on the chosen model parameters
shows some very interesting phenomenon.

The rest of this paper is organised as follows.
In Section \ref{sec:prel} we introduce the main definitions, notations and main identities
that will be used later. In Section \ref{sec:drawdown}, we analyse
the insurance contracts that are based only on drawdown (with and without cancellable feature).
In Section \ref{sec:drawup}, we add an additional possibility of stopping at the first drawup
(with and without cancellable feature). Some of our proofs are given in the Appendix.


\section{Preliminaries}\label{sec:prel}
We work on a complete filtered probability space $(\Omega,\mathcal{F},\mathbb{P})$
satisfying the usual conditions. We model a logarithm of risky underlying asset price $\log S_t$ by
a spectrally negative L\'evy process $X_t$; that is, $X_t$ is a stationary stochastic process with independent increments,
having right-continuous paths with left-hand finite limits and having only negative jumps (or not having jumps at all which means that $X_t$
is a  Brownian motion with linear drift). Any L\'evy process is associated with a triple $(\mu,\sigma,\Pi)$ by its characteristic function, as:
\begin{align}
\Psi(\phi):=-i\mu\phi+\frac{\sigma^2}{2}\phi^2+\int_{\mathbb{R}}(1-e^{i\phi y}+i\phi y\mathbbm{1}_{(|y|<1)})\Pi(\diff y),
\end{align}
where $\mu\in\mathbb{R}$, $\sigma\geq 0$ and a L\'evy measure $\Pi$ satisfies $\int_{\mathbb{R}}(1\wedge y^2)\Pi(\diff y)<\infty$.

In this paper we focus on two examples of L\'evy process $X_t$. We calculate all of the quantities explicitly and do whole numerical analysis
for them.
The first concerns a Black-Scholes market under which $X_t$ is the linear Brownian motion given by
\begin{align}
X_t=\mu t+\sigma B_t,\label{linearBrownian}
\end{align}
where $B_t$ is standard Brownian motion and if $\mathbb{P}$ is a martingale measure, then $\mu=r-\sigma^2/2$
for a risk-free interest rate $r$ and a volatility $\sigma >0$.
Obviously, $X_t$ in \eqref{linearBrownian} is a spectrally negative L\'evy process because it has no jumps.

In another classical example, we focus  on a Cram\'er-Lundberg process with the exponential jumps:
\begin{equation}\label{CL}
X_t=\hat{\mu} t-\sum_{i=1}^{N_t}{\eta_i},\end{equation}
where $\hat{\mu}=\mu - \int_{(0,1)}y\Pi(-\diff y)$, the sequence  $\{\eta_i\}_{\{i\geq 1\}}$ consists of i.i.d. exponentially distributed random variable with a parameter $\rho >0$
and $N_t$ is a Poisson process with an intensity $\beta>0$ independent of the sequence.

These two examples describe the most important features of L\'evy-type log-prices: their diffusive nature and their possible jumps.
As such, they may serve as core examples of the theory presented in this paper.

The main message of this paper is that fair premiums, prices of all contracts and all optimal stopping rules
can be expressed only in terms of two special functions, which are called the scale functions.
To define them properly, we introduce the Laplace exponent of $X_t$:
\begin{equation}
\psi (\phi ):=\log\mathbb{E}[e^{\phi X_1}]\label{Lapexp}
\end{equation}
which is well defined for $\phi\geq 0$ due to the absence of positive jumps.
Recall that, for $\mu\in\mathbb{R}$, $\sigma\geq 0$ and for a L\'evy measure $\Pi$, by L\'evy-Khintchine theorem:
\begin{equation}
\psi(\phi)=\mu\phi
+\frac{1}{2}\sigma^{2}\phi^{2}+\int_{(0,\infty)}\big(\mathrm e^{-\phi
y}-1+\phi y\mathbbm{1}_{(y<1)}\big)\Pi(-\diff y).\label{eq:exponent}
\end{equation}
Note that $\psi$ is zero at the origin, tends to infinity at infinity and it is strictly
convex. Therefore, we can properly define a right-inverse $\Phi:[0,\infty)\rightarrow [0,\infty)$
of the Laplace exponent $\psi$ given in the \eqref{Lapexp}. Thus,
\begin{equation*}
\Phi(r):=\sup\{\phi>0:\psi(\phi)=r\} \quad \textrm{and} \quad
\psi(\Phi(r))=r \quad \text{for all} \quad r \geq 0.
\end{equation*}
For $r\geq 0$ we define a continuous and strictly increasing function $W^{(r)}$ on $[0,\infty )$  with the Laplace transform given by:
\begin{align}
\int_0^\infty e^{-\phi y}W^{(r)}(y)\diff y=\frac{1}{\psi (\phi )-r}\quad \textrm{for} \quad \phi>\Phi(r).\label{W}
\end{align}
This is the so-called first scale function. From this definition, it also follows  that
$W^{(r)}$ is a non-negative function.
The second scale function is related to the first one via the following relationship:
\begin{align}
Z^{(r)}(\phi)=1+r\int_0^{\phi} W^{(r)}(y)\diff y. \label{Z}
\end{align}
In this paper, we assume that either the process $X_t$ has non-trivial Gaussian component --- that is, $\sigma >0$ (hence, it is of unbounded variation)
--- or it is of bounded variation and $\Pi(-\infty,-y)$ is continuous function for $y>0$. From \cite[Lem 2.4]{kyprianou},
it follows that
\begin{align}
W^{(r)}\in\mathcal{C}^1(\mathbb{R}_+)\label{Wprim}
\end{align}
for $\mathbb{R}_+=(0,\infty )$.
Moreover, under this assumptions the process $X_t$ has absolutely continuous transition density, that is, for any fixed $t>0$ the random variable $X_t$  is absolutely continuous.

\begin{Ex}
For linear Brownian motion \eqref{linearBrownian}
the Laplace exponent equals
\begin{align}
\psi(\phi)=\mu\phi+\frac{1}{2}\sigma^2\phi^2\nonumber
\end{align}
and, therefore,
the scale functions, for $\phi\geq 0$, are given as follows:
\begin{align}
W^{(r)}(\phi)&=\frac{1}{\Xi\sigma^2}\left(e^{(\Xi-\frac{\mu}{\sigma^2})\phi}-e^{(-\Xi-\frac{\mu}{\sigma^2})\phi}\right)=\frac{2}{\Xi\sigma^2}e^{-\frac{\mu}{\sigma^2}\phi}\sinh(\Xi\phi),\nonumber\\
Z^{(r)}(\phi)&=\frac{r}{\Xi\sigma^2}\left(\frac{e^{(\Xi-\frac{\mu}{\sigma^2})\phi}}{\Xi-\frac{\mu}{\sigma^2}}+\frac{e^{(-\Xi-\frac{\mu}{\sigma^2})\phi}}{\Xi+\frac{\mu}{\sigma^2}}\right)=e^{-\frac{\mu}{\sigma^2}\phi}\left(\cosh(\Xi\phi)+\frac{\mu}{\Xi\sigma^2}\sinh(\Xi\phi)\right),\nonumber
\end{align}
where
\begin{align}
\Xi=\frac{\sqrt{\mu^2+2r\sigma^2}}{\sigma^2}.\nonumber
\end{align}
\end{Ex}
\begin{Ex}

For the Cram\'er-Lundberg process \eqref{CL}
we have
\begin{align}
\psi(\phi)=\hat{\mu}\phi + \frac{\beta\rho}{\rho+\phi}-\beta\nonumber
\end{align}
and, hence, for $\phi\geq 0$,
\begin{align}
W^{(r)}(\phi) =&\frac{e^{\Phi(r)\phi}}{\psi^\prime(\Phi(r))}+\frac{e^{\zeta\phi}}{\psi^\prime(\zeta)},\nonumber\\
Z^{(r)}(\phi)=&1+r\frac{e^{\Phi(r)\phi}-1}{\Phi(r) \psi^\prime (\Phi(r) )}+r\frac{e^{\zeta\phi}-1}{\zeta\psi^\prime(\zeta)},\nonumber
\end{align}
where
\begin{align}
&\Phi(r) =\frac{1}{2\hat{\mu}}\left ((\beta+r-\hat{\mu}\rho)+\sqrt{(\beta+r-\hat{\mu}\rho)^2+4r\hat{\mu}\rho}\right ),\nonumber\\
&\zeta =\frac{1}{2\hat{\mu}}\left ((\beta+r-\hat{\mu}\rho)-\sqrt{(\beta+r-\hat{\mu}\rho)^2+4q\hat{\mu}\rho}\right ).\nonumber
\end{align}
\end{Ex}

Let us denote:
\begin{align}
\overline{X}_t=\sup_{s\leq t} X_s,\qquad \underline{X}_t=\inf_{s\leq t} X_s.\nonumber
\end{align}
In this paper, we analyse the insurance contracts related to the drawdown and drawup processes.
The classical definitions for these processes are as follows.
Drawdown is the difference between running maximum of the process and its current value.
Meanwhile, drawup is the difference between the process current value and its running minimum.
Without loss of generality, let us assume that $X_0=0$.
Additionally, one can allow that the drawdown and drawup processes start from some points $d\geq 0$ and $u\geq 0$, respectively.
That is,
\begin{align}
&D_t=\overline{X}_t\vee d-X_t,\quad D_0=d,\nonumber\\
&U_t=X_t-\underline{X}_t\wedge (-u),\quad U_0=u.\nonumber
\end{align}
Thus, the above values $d$ and $-u$ may be interpreted as the historical maximum and historical minimum of process $X$.
In daily practice, zero level of $X_0$ might be treated as the present position of log-prices of the asset that we work with.
In this case, the above interpretations of $d$ and $-u$ are even more clear.

The following first passage times of drawdown and drawup processes are crucial for further work, respectively:
\begin{align}
&\tau_{D}^+(a):=\inf\{t\geq 0 : D_t> a\},\quad\tau_{D}^-(a):=\inf\{t\geq 0 : D_t> a\},\nonumber\\
&\tau_{U}^+(b):=\inf\{t\geq 0 : U_t> b\},\quad\tau_{U}^-(b):=\inf\{t\geq 0 : U_t> b\},\nonumber
\end{align}
for some $a,b>0$.

Later, for fixed $d,\ u$ we will use the following notational convention:
\begin{align}
\begin{array}{ll}
\mathbb{P}_{|d}\left[\cdot\right]:=\mathbb{P}\left[\cdot|X_0=0, D_0=d\right],&\mathbb{P}_{x|d}\left[\cdot\right]:=\mathbb{P}\left[\cdot|X_0=x,D_0=d\right],\\
\mathbb{P}_{|d|u}\left[\cdot\right]:=\mathbb{P}\left[\cdot|X_0=0, D_0=d,\ U_0=u\right],&\mathbb{P}_{x|d|u}\left[\cdot\right]:=\mathbb{P}\left[\cdot|X_0=x,\ D_0=d,\ U_0=u\right].
\end{array}\nonumber
\end{align}
Finally, we denote $\mathbb{P}_{x}\left[\cdot\right]:=\mathbb{P}\left[\cdot|X_0=x\right]$, with $\mathbb{P}=\mathbb{P}_{0}$ and $\mathbb{E}_{|d}, \mathbb{E}_{x|d},\mathbb{E}_{|d|u}, \mathbb{E}_{x|d|u}, \mathbb{E}_{x}, \mathbb{E}$ will be corresponding expectations to the above measures.
We will also use the following notational convention: $\mathbb{E}[\cdot\ \mathbbm{1}_{(A)}]=\mathbb{E}[\cdot ;A]$.

The seminal observation for the fluctuation of L\'evy processes
is the fact the scale functions \eqref{W} and \eqref{Z}
are used in solving so-called exit problems given by:
\begin{align}
&\mathbb{E}_{x}\left[e^{-r\tau^+_a};\tau^+_a<\tau^-_0\right]=\frac{W^{(r)}(x)}{W^{(r)}(a)},\label{twosided1}\\
&\mathbb{E}_{x}\left[e^{-r\tau^-_0};\tau^-_0<\tau^+_a\right]=Z^{(r)}(x)-Z^{(r)}(a)\frac{W^{(r)}(x)}{W^{(r)}(a)},\label{twosided2}
\end{align}
where $x\leq a$, $r\geq 0$ and
\begin{align}
\tau^+_a:=\inf\{t\geq 0 : X_t\geq a\},\qquad\tau^-_a:=\inf\{t\geq 0 : X_t< a\}\nonumber
\end{align}
are the first passage times of the process $X_t$.
We finish this section with the formula (given in Mijatovi\'{c} and Pistorius \cite[Thm. 3]{Mijatovic&Pistorius})
that identifies the joint law of $\{\tau_{U}^+(b), \overline{X}_{\tau_{U}^+(b)}$, $\underline{X}_{\tau^+_{U}(b)}\}$, for $r,u,v\geq 0$:
\begin{align}
\mathbb{E}\left[e^{-r\tau^+_{U}(b)+u\underline{X}_{\tau^+_{U}(b)}};\overline{X}_{\tau^+_{U}(b)}<v\right]=&e^{ub}\frac{1+(r-\psi (u))\int_0^{b-v}e^{-uy}W^{(r)}(y)dy}{1+(r-\psi (u))\int_0^{b}e^{-uy}W^{(r)}(y)dy}\nonumber\\
&-e^{-u(b-v)}\frac{W^{(r)}(b-v)}{W^{(r)}(b)}.\label{m2}
\end{align}

\section{Drawdown insurance contract}\label{sec:drawdown}
\subsection{Fair premium}
In this section we consider the insurance contract in which
a protection buyer pays constant premium $p>0$ continuously until the drawdown of size $a>0$ occurs.
In return, he/she receives the reward $\alpha(D_{\tau_{D}^+(a)})$ that depends on the value of the drawdown process at this moment of time.
It is natural to assume that $\alpha(\phi)=0$ for $\phi<a$.
Let $r\geq 0$ be the risk-free interest rate. The price of the first, basic contract that we consider
in this paper equals the discounted value of the future cash-flows:
\begin{align}\label{f}
f(d,p) := \mathbb{E}_{|d}\left[-\int_0^{\tau_{D}^+(a)}e^{-rt}p\diff t+\alpha(D_{\tau_{D}^+(a)})e^{-r\tau_{D}^+(a)}\right].
\end{align}
In this contract, the investor wants to protect herself/himself from the asset price $S_t=e^{X_t}$ falling down from the previous maximum
more than fixed level $e^a$ for some fixed $a>0$.
In other words, she/he believes that even if the price will go up again after the first drawdown of size $e^a$
it will not bring her/him sufficient profit. Therefore, she/he is
ready to take this type of contract to reduce loss by getting $\alpha(D_{\tau_{D}^+(a)})$
at the drawdown epoch.

Note that we define the drawdown contract value \eqref{f} from the investor's position.
This represents the investor's average profit when the value is positive or loss when this value is negative.
Thus, the contract is unprofitable for the insurance company in the first case and for the investor in the second case.
The only fair solution for both sides situation is when the contract value equals zero.
Obviously, this means that this is the premium $p$ under which the contract should be constructed
or it can serve as the basic reference economical premium.

\begin{Def}\label{p*definition}
The premium $p$ is said to be \emph{\textbf{fair}} when contract value at its beginning is equal to $0$.
We denote this fair premium by $p^*$. We will add argument of initial drawdown (and later drawup) to
underline its dependence on the initial conditions, e.g. writing $p^*(d)$ or $p^*(d,u)$.
\end{Def}

The main of goal of this section and the paper is identifying
the fair premium $p^*$ for the first contract \eqref{f} under spectrally negative L\'evy-type market.

Let us start from the basic observation that:
\begin{equation}\label{f1}
f(d,p)= \frac{p}{r}\xi(d)-\frac{p}{r}+\Xi(d),
\end{equation}
where
\begin{align}
\xi(d):=&\mathbb{E}_{|d}\left[e^{-r\tau_{D}^+(a)}\right],\label{xi}\\
\Xi(d):=&\mathbb{E}_{|d}\left[e^{-r\tau_{D}^+(a)}\alpha(D_{\tau_{D}^+(a)})\right]\label{Xi}
\end{align}
are the Laplace transform of $\tau_{D}^+ (a)$ and the discounted reward function, respectively.
Note that $\xi\in[0,1]$ is well defined.

Moreover, from now on we assume that for all $d\geq 0$
\begin{align}
\Xi(d)<\infty,\label{warXi}
\end{align}
for $\Xi$ to be well defined.
Observe that \eqref{warXi} holds true, such as for the bounded reward function $\alpha$.
To price the contract \eqref{f}, we start by identifying
the crucial functions $\xi$ and $\Xi$.

\begin{Prop}\label{Ksi}
The value of the drawdown contract \eqref{f} equals \eqref{f1} for $\xi$ and $\Xi$ given by:
\begin{align}
\xi (d) = Z^{(r)}(a-d)-rW^{(r)}(a-d)\frac{W^{(r)}(a)}{W'^{(r)}(a)}\nonumber
\end{align}
and
\begin{align}
\Xi(d)=&\int_0^a\!\!\!\int_0^\infty\alpha(a+h)\left[W^{(r)}(a-d)\frac{W'^{(r)}(a-z)}{W'^{(r)}(a)}-W^{(r)}(a-d-z)\right]\Pi(-z-\diff h)\diff z\nonumber\\
&+\alpha(a)\frac{\sigma^2}{2}\left[W'^{(r)}(a-d)-W^{(r)}(a-d)\frac{W''^{(r)}(a)}{W'^{(r)}(a)}\right],\nonumber
\end{align}
where $\Pi$ is the L\'evy measure of underlying process $X_t$ defined formally in \eqref{eq:exponent}.
\end{Prop}
\begin{proof}
See Appendix.
\end{proof}

Note that $\tau_D^+(a)<\infty$ a.s. which follows from  Proposition \ref{Ksi} by taking $r=0$ in $\xi(d)$.

From \eqref{f1} and using Proposition \ref{Ksi}, we can derive the following theorem.
\begin{theorem}
For the contract \eqref{f} the fair premium equals:
\begin{equation}
p^* (d) = \frac{r\Xi (d)}{1-\xi (d)}.\label{p*}
\end{equation}
\end{theorem}
\addtocounter{Ex}{-2}
\begin{Ex}[continued]
The linear Brownian motion given in \eqref{linearBrownian} is a continuous process and, therefore, $D_{\tau_D^+(a)}=a$. The paid reward is always equal to $\alpha(a):=\alpha$, which corresponds to the results for constant reward function in \cite{ZP&JT}. The value function equals:
\begin{align}
f(d,p)=\left(\frac{p}{r}+\alpha\right)e^{-\frac{\mu}{\sigma^2}(a-d)}\frac{\Xi\cosh(\Xi d)-\frac{\mu}{\sigma^2}\sinh(\Xi d)}{\Xi\cosh(\Xi a)-\frac{\mu}{\sigma^2}\sinh(\Xi a)}-\frac{p}{r}\nonumber
\end{align}
and the fair premium $p^*$ is given by:
\begin{align}
p^*(d)=\frac{r\alpha\cdot(\Xi\cosh(\Xi d)-\frac{\mu}{\sigma^2}\sinh(\Xi d))}{\Xi(e^{\frac{\mu}{\sigma^2}(a-d)}\cosh(\Xi a)-\cosh(\Xi d))-\frac{\mu}{\sigma^2}(e^{\frac{\mu}{\sigma^2}(a-d)}\sinh(\Xi a)-\sinh(\Xi d))}.\nonumber
\end{align}

In Figure \ref{fig:p*bm},
we demonstrate the $p^*$ value for the drawdown insurance contract
for the Black-Scholes market for various values of $\alpha$. We choose
the following parameters: $r=0.01, \mu=0.03, \sigma=0.4, a=10$.

In Figure \ref{fig:fbm}, we present the contract value $f$ for various premiums $p$.
We choose the same parameters as above and fixed $\alpha=100$.
Finally, for the premium rates, we take the fair premiums
when the initial drawdown equals $d=0$, $d=5$, $d=6$ and $d=7$, respectively.

\begin{figure}[!ht]
\centering
\includegraphics[width=0.6\textwidth]{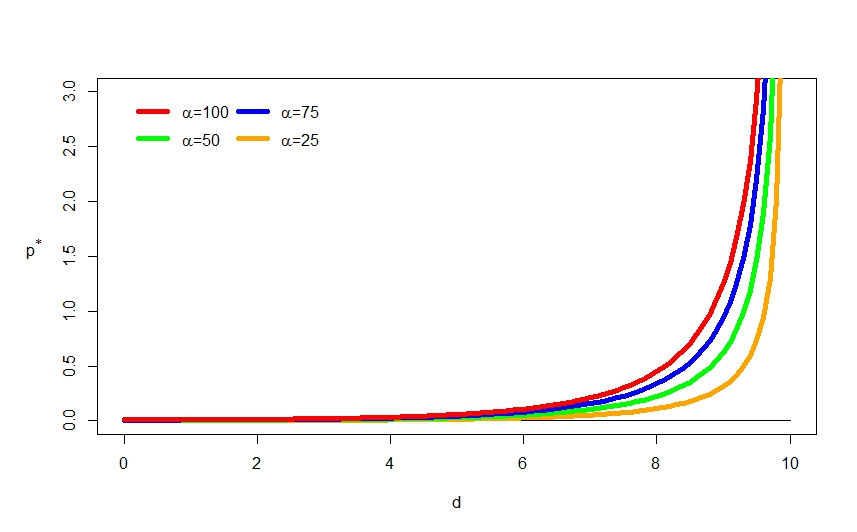}
\caption{\footnotesize{The $p^*$ values for drawdown contract for the linear Brownian motion and various reward values $\alpha$. Parameters: $r=0.01, \mu=0.03, \sigma=0.4, a=10$.}}
\label{fig:p*bm}
\end{figure}
\begin{figure}[!ht]
\centering
\includegraphics[width=0.6\textwidth]{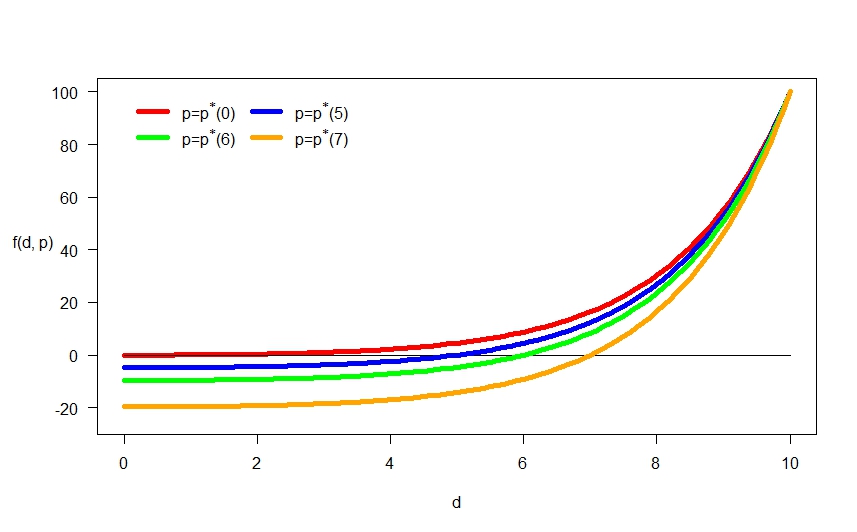}
\caption{\footnotesize{The contract value $f$ for the drawdown contract for the linear Brownian motion and various premiums $p$.
Parameters: $r=0.01, \mu=0.03, \sigma=0.4, a=10, \alpha=100$.}}
\label{fig:fbm}
\end{figure}
\end{Ex}
\begin{Ex}[continued]
Let us consider now Cram\'er-Lundberg process defined in \eqref{CL}. From \cite{ZP&JT} we have
\begin{align}
\xi (d) =& 1+r\frac{e^{\Phi(r) (a-d)}-1}{\Phi(r)\psi^\prime(\Phi(r) )}+r\frac{e^{\zeta (a-d)}-1}{\zeta\psi^\prime(\zeta)}-r\left ( \frac{e^{\Phi(r) (a-d)}}{\psi^\prime(\Phi(r) )}+\frac{e^{\zeta (a-d)}}{\psi^\prime(\zeta)}\right ) \nonumber \\
&\cdot\frac{\psi^\prime(\zeta)e^{\Phi(r) a}+\psi^\prime(\Phi(r) )e^{\zeta a}}{\Phi(r) \psi^\prime(\zeta)e^{\Phi(r) a}+\zeta\psi^\prime(\Phi(r) )e^{\zeta a}}=: c_0 + c_{\Phi(r)}e^{\Phi(r) (r)(a-d)}+c_{-\zeta}e^{-\zeta (a-d)},\nonumber
\end{align}
where
\begin{align}
&c_0=1-\frac{r}{\Phi(r) \psi^\prime(\Phi(r) )}-\frac{r}{\zeta\psi^\prime(\zeta)},\nonumber\\
&c_{\Phi(r)}=\frac{r}{\Phi(r) \psi^\prime(\Phi(r) )}-\frac{r}{\psi^\prime(\Phi(r) )}\frac{\psi^\prime(\zeta)e^{\Phi(r) a}+\psi^\prime(\Phi(r) )e^{\zeta a}}{\Phi(r) \psi^\prime(\zeta)e^{\Phi(r) a}+\zeta\psi^\prime(\Phi(r) )e^{\zeta a}},\nonumber\\
&c_{-\zeta}=\frac{r}{\zeta\psi^\prime(\zeta)}-\frac{r}{\psi^\prime(\zeta)}\frac{\psi^\prime(\zeta)e^{\Phi(r) a}+\psi^\prime(\Phi(r) )e^{\zeta a}}{\Phi(r) (r)\psi^\prime(\zeta)e^{\Phi(r) a}+\zeta\psi^\prime(\Phi(r) )e^{\zeta a}}.\nonumber
\end{align}
For Cram\'er-Lundberg model we have $D_{\tau_D^+(a)}>a$ because of the jump's presence.
Let us consider the general reward function $\alpha(d)$ for $d\geq a$. Then
\begin{align}
\Xi(d)=&\int_0^a\!\!\!\int_0^\infty \alpha (a+h)\left[W^{(r)}(a-d)\frac{W'^{(r)}(a-z)}{W'^{(r)}(a)}-W^{(r)}(a-d-z)\right]\beta\rho e^{-\rho z}e^{-\rho h}\diff h\diff z\nonumber\\
=&\int_0^\infty \alpha (a+h) \rho e^{-\rho h}\diff h\int_0^a \beta e^{-\rho z}\left[W^{(r)}(a-d)\frac{W^{\prime (r)}(a-z)}{W^{\prime (r)}(a)}-W^{(r)}(a-d-z)\right]\diff z\nonumber\\
=&\E\left[\alpha (a+\mathbf{e}_{\rho})\right]\xi(d),\nonumber
\end{align}
where $\mathbf{e}_{\rho}$ is exponentially distributed random variable with parameter $\rho$.
Therefore, the value function and the fair premium are given by
\begin{align}
&f(d,p)=\left(\frac{p}{r}+\E\left[\alpha (a+\mathbf{e}_{\rho})\right]\right)\xi(d)-\frac{p}{r},\nonumber\\
&p^*(d)=\frac{r\E\left[\alpha (a+\mathbf{e}_{\rho})\right]\xi(d)}{1-\xi(d)}.\nonumber
\end{align}

From the above representation of $\Xi$ we can easily check the condition \eqref{warXi}. For example, let $\alpha(d)=\omega e^{\kappa d}$, where $\omega$ and $\kappa$ are some constants. For this reward function, condition \eqref{warXi} is satisfied when:
\begin{align}
\mathbb{E}\left[\alpha(a+\mathbf{e}_{\rho})\right]=\omega e^{\kappa a}\mathbb{E}\left[e^{\kappa\mathbf{e}_{\rho}}\right]=\omega e^{\kappa a}\rho\int_0^\infty e^{-(\rho-\kappa)x}\diff x<\infty.\nonumber
\end{align}
This inequality is satisfied when $\kappa<\rho$ and $\omega<\infty$.
Another example can be found in the linear reward function $\alpha(d)=\alpha_1 +\alpha_2 d$ for some constants $\alpha_1$ and $\alpha_2.$ In this case
\begin{align}
\mathbb{E}\left[\alpha(a+\mathbf{e}_{\rho})\right]=\alpha_1+\alpha_2 a +\alpha_2\mathbb{E}\left[\mathbf{e}_{\rho}\right]=\alpha\left(a+\frac{1}{\rho}\right).\nonumber
\end{align}
Thus, the linear reward function satisfied condition \eqref{warXi} always when $\alpha_1, \alpha_2<\infty$.
In particular, for the linear reward function $\alpha (\cdot )$ we get:
\begin{align}
f(d,p)&=\left(\frac{p}{r}+\alpha\left(a+\frac{1}{\rho}\right)\right)\left(c_0 + c_{\Phi(r)}e^{\Phi(r) (a-d)}+c_{-\zeta}e^{\zeta (a-d)}\right)-\frac{p}{r},\nonumber\\
p^*(d)&=\frac{r\alpha\left(a+\frac{1}{\rho}\right)\cdot(c_0 + c_{\Phi(r)}e^{\Phi(r) (a-d)}+c_{-\zeta}e^{\zeta (a-d)})}{1-c_0 - c_{\Phi(r)}e^{\Phi(r) (a-d)}-c_{-\zeta}e^{\zeta (a-d)}}.\nonumber
\end{align}

For the Cram\'er-Lundberg model the we present in Figures \ref{fig:p*cl} and \ref{fig:fcl},
the $p^*$ value for various linear reward functions $\alpha(d)$ and
the contract value $f(d,p)$ for two premiums: $p=p^*(0)$ (which means that we start at the temporary maximum asset price)
and for $p=p^*(7.5)$, respectively.

\begin{figure}[!ht]
\centering
\includegraphics[width=0.6\textwidth]{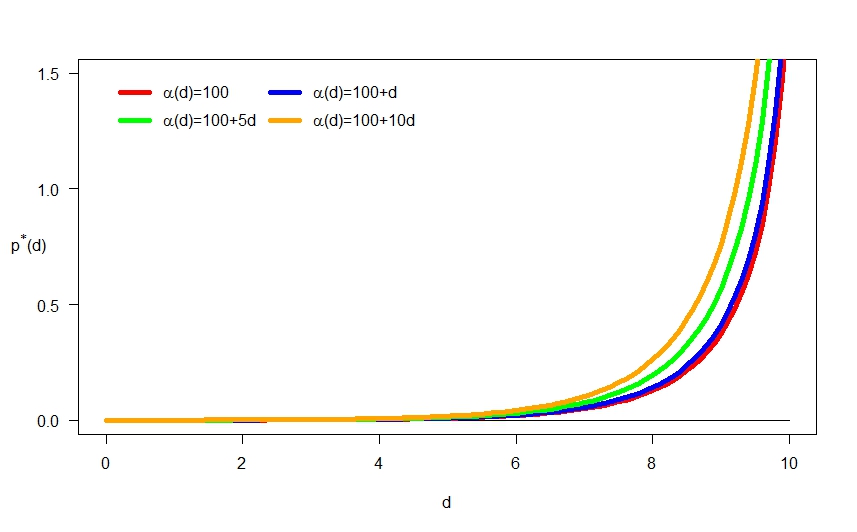}
\caption{\footnotesize{The $p^*$ values for the drawdown contract for Cram\'er-Lundberg model and various linear reward functions $\alpha(\cdot)$.
Parameters: $r=0.01, \hat{\mu}=0.05, \beta=0.1, \rho=2.5, a=10$.}}
\label{fig:p*cl}
\end{figure}
\begin{figure}[!ht]
\centering
\includegraphics[width=0.6\textwidth]{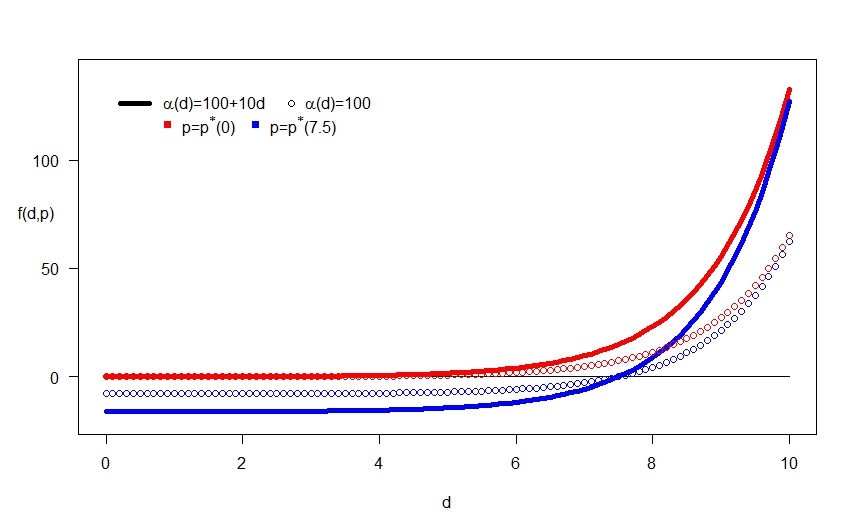}
\caption{\footnotesize{The contract value $f$ for the drawdown contract for Cram\'er-Lundberg model for various premiums $p$ and linear reward functions $\alpha(\cdot)$.
Parameters: $r=0.01, \hat{\mu}=0.05, \beta=0.1, \rho=2.5, a=10$.}}
\label{fig:fcl}
\end{figure}
\end{Ex}

\subsection{Cancellable feature}
We now extend the previous contract by adding a cancellable feature.
In other words, we give the investor a right to terminate the contract at any time prior to a pre-specified drawdown of size $a>0$.
To terminate, the contract she/he pays a fee $c(D_\tau)$ that depends on the level of the drawdown at the termination time.
It is intuitive to assume that the penalty function $c(\cdot )$ is non-increasing.
Indeed, if the value of drawdown is increasing, then the investor losses more money and the fee should decrease.
For simplicity of future calculation, let also $c$ be in $C^2(\mathbb{R}_+)$.
The value of the contract that we consider here equals:
\begin{align}
F(d,p) := \sup\limits_{\tau\in\mathcal{T}}\mathbb{E}_{|d}\Bigg[-\int_0^{\tau_{D}^+(a)\wedge\tau}e^{-rt}p \diff t - c(D_\tau)e^{-r\tau}\mathbbm{1}_{(\tau<\tau_{D}^+(a))} + \alpha(D_{\tau_{D}^+(a)}) e^{-r\tau_{D}^+(a)}\mathbbm{1}_{(\tau_{D}^+(a)\leq\tau)}; \tau<\infty\Bigg],\label{F}
\end{align}
where $\mathcal{T}$ is a family of all $\mathcal{F}_t$-stopping times.

One of the main goals of this paper is to identify the optimal stopping rule $\tau^*$ that realises the price $F(d,p)$. We start from the simple observation.
\begin{Prop}\label{FDecomp}
The cancellable drawdown insurance value admits the following decomposition:
\begin{align}\label{Fdecomposition}
F(d,p)=f(d,p)+G(d,p),
\end{align}
where
\begin{align}
&G(d,p):=\sup\limits_{\tau\in\mathcal{T}}g_\tau(d,p),\label{G}\\
&g_\tau(d,p):=\mathbb{E}_{|d}\left[e^{-r\tau}\tilde{f}(D_{\tau},p);\ \tau<\tau_{D}^+(a),\ \tau<\infty\right]\label{gtaub}\\
&\tilde{f}(d,p):=-f(d,p)-c(d)\label{tildef}
\end{align}
for $f$ defined in (\ref{f}).
\end{Prop}
\begin{proof}
Using $\mathbbm{1}_{(\tau\geq\tau_{D}^+(a))}=1-\mathbbm{1}_{(\tau <\tau_{D}^+(a))}$ in (\ref{F}) we obtain:
\begin{align}
&F(d,p)=\mathbb{E}_{|d}\left[ -\int_0^{\tau_{D}^+(a)}e^{-rt}p\diff t +\alpha(D_{\tau_{D}^+(a)}) e^{-r\tau_{D}^+(a)}\right] \nonumber \\
&\qquad +\sup\limits_{\tau\in\mathcal{T}}\mathbb{E}_{|d}\left[\int_{\tau}^{\tau_{D}^+(a)}e^{-rt}p \diff t\mathbbm{1}_{\left(\tau<\tau_D^+(a)\right)} - \alpha(D_{\tau_{D}^+(a)}) e^{-r\tau_{D}^+(a)}\mathbbm{1}_{\left(\tau<\tau_{D}^+(a)\right)} - c(D_\tau)e^{-r\tau}\mathbbm{1}_{\left(\tau <\tau_{D}^+(a)\right)}; \tau<\infty\right]. \nonumber
\end{align}
Note that the first term does not depend on $\tau$. The second term depends on $\tau$ only in event $\{\tau <\tau_D^+(a)\}$. Then, by using a strong Markov property we get:
\begin{align}
&F(d,p)=f(d,p)\nonumber\\
&\qquad + \sup\limits_{\tau\in\mathcal{T}}\mathbb{E}_{|d}\left[\int_{\tau}^{\tau_D^+(a)}e^{-rt}p \diff t\mathbbm{1}_{\left(\tau<\tau_D^+(a)\right)} - \alpha(D_{\tau_D^+(a)}) e^{-r\tau_D^+(a)}\mathbbm{1}_{\left(\tau<\tau_D^+(a)\right)} - c(D_\tau)e^{-r\tau}\mathbbm{1}_{\left(\tau <\tau_D^+(a)\right)};\ \tau<\infty\right] \nonumber \\
&=f(d,p)+\sup\limits_{\tau\in\mathcal{T}}\mathbb{E}_{|d}\left[e^{-r\tau}\mathbb{E}_{|D_\tau}\left[\int_{0}^{\tau_D^+(a)}e^{-rt}p\diff t - \alpha(D_{\tau_D^+(a)}) e^{-r\tau_D^+(a)} - c(D_\tau);\ \tau<\infty\right];\ \tau <\tau_D^+(a),\ \tau<\infty\right]. \nonumber
\end{align}
This completes the proof.
\end{proof}

To determine the optimal cancellation strategy for out contract, it is sufficient to solve the optimal stopping problem represented by the second increment of \eqref{Fdecomposition},
that is used to identify function $G(d,p)$. We use the ``guess and verify'' approach. This means that we first guess the candidate stopping rule and then we
verify that this is truly the optimal stopping rule
using the Verification Lemma given below.
\begin{lem}\label{war}
Let $\Upsilon_t$ be a right-continuous process living in some Borel state space $\mathbb{B}$
killed at some $\mathcal{F}^\Upsilon_t$-stopping time $\tau_0$, where $\mathcal{F}^\Upsilon_t$ is a right-continuous natural filtration of $\Upsilon$.
Consider the following stopping problem:
\begin{equation}\label{BVP}
v(\phi)=\sup\limits_{\tau\in\mathcal{T}^\Upsilon}\mathbb{E}\left[e^{-r\tau}V(\Upsilon_\tau)|\Upsilon_0=\phi\right]
\end{equation}
for some function $V$ and the family $\mathcal{F}^\Upsilon_t$-stopping times $\mathcal{T}^\Upsilon$.
Assume that
\begin{equation}
\mathbb{P}(\lim_{t\rightarrow \infty}e^{-rt}V(\Upsilon_t)<\infty|\Upsilon_0=\phi)=1.
\end{equation}
The pair $(v^*,\tau^*)$ is a solution of stopping problem (\ref{BVP}) if
for
$$v^*(\phi):=\mathbb{E}\left[e^{-r\tau^*}V(\Upsilon_{\tau^*})|\Upsilon_0=\phi\right]$$
we have
\begin{enumerate}[i)]
\item $v^*(\phi)\geq V(\phi)$ for all $\phi\in\mathbb{B}$ \label{i}
\item the process $e^{-rt}v^*(\Upsilon_t)$ is a right-continuous supermartingale. \label{ii}
\end{enumerate}
\end{lem}
\begin{proof}
The proof follows the same arguments as the proof of \cite[Lem. 9.1, p. 240]{KIntr}; see also \cite[Th. 2.2, p. 29]{Peskir}.
\end{proof}

Using above Verification Lemma \ref{war} we prove that the first passage time of drawdown process below some level $\theta$ is the
optimal stopping time for \eqref{G} (hence, also for \eqref{Fdecomposition}); that is,
\begin{equation}\label{tau*pierwsze}
\tau^*:=\tau_D^-(\theta)\in\mathcal{T}.\end{equation}
for some optimal $\theta^*\in (0,a)$.

For the stopping rule (\ref{tau*pierwsze}), we consider two cases: when $d>\theta$ and when $d\leq\theta$; that is, we decompose
$g_{\tau_D^-(\theta)}(d,p)$ as follows:
\begin{align}
g_{\tau_D^-(\theta)}(d,p)=g_{\tau_D^-(\theta)}(d,p)\mathbbm{1}_{(d>\theta)}+g_{\tau_D^-(\theta)}(d,p)\mathbbm{1}_{(d\leq\theta)}:=g_>(d,p,\theta)\mathbbm{1}_{(d>\theta)}+g_<(d,p,\theta)\mathbbm{1}_{(d\leq\theta)}.\nonumber
\end{align}
If $d\leq\theta$, then we have situation when investor should stop the contract immediately. In this case we have,
\begin{eqnarray}
g_<(d,p,\theta)=\tilde{f}(d,p).\label{g=f}
\end{eqnarray}
To analyse the complimentary case of $d>\theta$, observe now that, the choice of class of stopping time \eqref{tau*pierwsze}
implies that $D_{\tau^*}=D_{\tau_D^-(\theta)}=\theta$ since $X$ is spectrally negative and it goes upward continuously. This implies that:
\begin{align}
g_>(d,p,\theta)=\E_{|d}\left[e^{-r\tau_D^-(\theta)}\tilde{f}(\theta,p);\ \tau_D^-(\theta)<\tau_D^+(a)\right]=\tilde{f}(\theta,p)\frac{W^{(r)}(a-d)}{W^{(r)}(a-\theta)}.\label{g}
\end{align}
From this structure, it follows that it is optimal to never terminate the contract earlier if $\tilde{f}(\theta,p)<0$ for all $\theta$.
To eliminate this trivial case, we assume from now on that there exist at least one $\theta_0\geq 0$ which satisfies
\begin{equation}\label{mainzalozenia}
\tilde{f}(\theta_0,p)>0.
\end{equation}

Recall that by \eqref{Fdecomposition} the optimal level $\theta^*$ maximises the value function $F(d,p)$ given in \eqref{F} and, hence,
it also maximises the value function $G(d,p)$ \eqref{G}.
Thus, we have to choose $\theta^*$ as follows:
\begin{align}\label{teta*}
\theta^*=\inf\left\{ \theta\in (0,a):\ g_>(d,p,\varsigma)\leq g_>(d,p,\theta) \ \forall \varsigma\geq 0\  \right\}.
\end{align}
\begin{lem}
Assume that \eqref{mainzalozenia} holds.
Then, there exists $\theta^*$ given by \eqref{teta*} and, additionally, $\theta^*$ does not depend on starting position of the drawdown $d$.
\end{lem}
\begin{proof}
At the beginning, note that by \eqref{f1}, Proposition \ref{Ksi}, assumption \eqref{Wprim} and \eqref{tildef}, the function $\tilde{f}$ is continuous.
Moreover, $\tilde{f}(a,p)=-\alpha (a)<0$ since for the initial drawdown $d=a$ the insured drawdown level is achieved at the beginning and the reward has to be paid immediately.
Consider now two cases.
If $\tilde{f}(0,p)<0$, then, by assumption \eqref{mainzalozenia}, there exist $\theta_0$ such that $\tilde{f}(\theta_0,p)>0$.
This implies existence of $\theta^*$ defined in \eqref{teta*}.

On the other hand, if $\tilde{f}(0,p)\geq 0$, then using \eqref{g} we have:
\begin{align}
\frac{\partial}{\partial\theta}g_>(d,p,\theta)&=\frac{\partial}{\partial\theta}\tilde{f}(\theta,p)\frac{W^{(r)}(a-d)}{W^{(r)}(a-\theta)}+\tilde{f}(\theta,p)\frac{W^{(r)}(a-d)W^{\prime(r)}(a-\theta)}{(W^{(r)}(a-\theta))^2}\nonumber\\
&=\frac{W^{(r)}(a-d)}{W^{(r)}(a-\theta)}\left(\frac{\partial}{\partial\theta}\tilde{f}(\theta,p)+\tilde{f}(\theta,p)\frac{W^{\prime (r)}(a-\theta)}{W^{(r)}(a-\theta)}\right).\label{gprime_theta}
\end{align}
Note that by \eqref{xi}-\eqref{Xi} we have
$\xi^\prime (0)=0$ and $\Xi^\prime (0)=0$. Thus
$$\frac{\partial}{\partial\theta}\tilde{f}(\theta ,p)|_{\theta=0}=-\frac{p}{r}\xi^\prime (0)-\Xi^\prime (0)-c^\prime (0)=-c^\prime (0)>0.$$
Recall that scale function $W^{(r)}$ given in \eqref{W} is non-negative. 
Thus, $\frac{\partial}{\partial\theta}g_>(d,p,\theta)|_{\theta=0}>0$ as each term in \eqref{gprime_theta} is positive.
We can now conclude that there exist local maximum $\theta^*>0$. Note that to find $\theta^*$ it suffices
to make expression in the brackets of \eqref{gprime_theta} equal zero and, therefore,
$\theta^*$ does not depend on $d$.
\end{proof}

The choice of the stopping rule \eqref{tau*pierwsze} ensures that the ``continuous fit'' condition holds:
\begin{align}
g_>(d,p,\theta^*)|_{d\searrow\theta^*}=g_<(d,p,\theta^*)|_{d\nearrow\theta^*}\qquad\textrm{\textit{(continuous fit)}}.\label{continuous_fit}
\end{align}
Indeed, the continuous fit is satisfied for any $\theta\in (0,a)$ from the continuity of scale function $W^{(r)}$ and the
definition of $g_>(d,p,\theta)$ given in \eqref{g}. Moreover, if we choose $\theta=\theta^*$ defined in \eqref{teta*}, then we have a ``smooth fit'' property satisfied:
\begin{align}
\frac{\partial}{\partial d}g_>(d,p,\theta^*)|_{d\searrow\theta^*}=\frac{\partial}{\partial d}g_<(d,p,\theta^*)|_{d\nearrow\theta^*}\qquad\textrm{\textit{(smooth fit)}}.\label{smooth_fit}
\end{align}
The {\it smooth fit} for $\theta^*$ follows from its definition \eqref{teta*} and from equation \eqref{gprime_theta}.
Indeed, since $\theta^*$ maximises $g_>(d,p,\theta)$ (given in \eqref{g}) with respect to $\theta$ it must be a root of \eqref{gprime_theta}.
Then, by equation \eqref{g=f} we have,
\begin{align}
\frac{\partial}{\partial d}g_>(d,p,\theta^*)|_{d=\theta^*}=-\tilde{f}(\theta^*,p)\frac{W^{\prime (r)}(a-\theta^*)}{W^{(r)}(a-\theta^*)}=\frac{\partial}{\partial d}\tilde{f}(d,p)|_{d=\theta^*}=\frac{\partial}{\partial d}g_<(d,p,\theta^*)|_{d=\theta^*}.\nonumber
\end{align}
Note that from the above equalities we see that the {\it smooth fit} conditions is satisfied whenever the scale function $W^{(r)}$ has continuous derivative. Thus, the {\it smooth fit} condition is always fulfilled in our framework, as we assumed in \eqref{Wprim}.

We will now verify that the stopping time $\tau^*$ given in \eqref{tau*pierwsze} is indeed optimal.
The proof of this fact is based on Verification Lemma \ref {war} and we start all of our considerations from the following lemma.
\begin{lem}\label{supermtg}
By adding superscript $c$ to some process we denote its continuous part.
Let $X_t$ be spectrally negative L\'evy process and define two disjoint regions:
\begin{align}
&C=\{d\in\mathbb{R}_+ :\ d>\theta^*\},\nonumber\\
&S=\{d\in\mathbb{R}_+ :\ 0\leq d<\theta^*\}\nonumber
\end{align}
for some $\theta^*\in (0,a)$.
Let $p(t,d,u)$ be $C^{1,2,2}$-function on $\overline{C}$ and $\overline{S}$.
We denote $$p_C(t,d,u):=p(t,d,u)\mathbbm{1}_{(d\in C)}\qquad\text{and}\qquad p_S(t,d,u):=p(t,d,u)\mathbbm{1}_{(d\in S)}.$$
For any function $\hat{p}$ we denote:
\begin{align}
I(\hat{p})(t,D_t, U_t):=&\hat{p}(0,D_0, U_0)+\int_0^t\frac{\partial}{\partial s}\hat{p}(s,d,u)|_{(s,d,u)=(s,D_s, U_s)}\diff s\nonumber\\
&+\int_0^t \frac{\partial}{\partial d}\hat{p}(s,d,u)|_{(s,d,u)=(s,0, U_{s-})}\diff \overline{X}^c_s
-\int_0^t \frac{\partial}{\partial d}\hat{p}(s,d,u)|_{(s,d,u)=(s,D_{s-}, U_{s-})}\diff X_s^c\nonumber\\
&+\int_0^t \frac{\partial}{\partial u}\hat{p}(s,d,u)|_{(s,d,u)=(s,D_{s-}, U_{s-})}\diff X_s^c-\int_0^t \frac{\partial}{\partial u}\hat{p}(s,d,u)|_{(s,d,u)=(s,D_{s-}, 0)}\diff \underline{X}^c_s\nonumber\\
&+\frac{1}{2}\int_0^t \left(\frac{\partial^2}{\partial d^2}
+2\frac{\partial^2}{\partial d\partial u}+\frac{\partial^2}{\partial u^2}\right)\hat{p}(s,d,u)|_{(s,d,u)=(s,D_{s-}, U_{s-})}\diff [X]^c_s\nonumber\\
&+\sum\limits_{0<s\leq t}\bigg(\hat{p}(s,D_s, U_s)-\hat{p}(s,D_{s-}, U_{s-})\nonumber\\&\qquad -\frac{\partial}{\partial d}\hat{p}(s,d,u)|_{(s,d,u)=(s,D_s, U_s)}\Delta D_s
-\frac{\partial}{\partial u}\hat{p}(s,d,u)|_{(s,d,u)=(s,D_s, U_s)}\Delta U_s\bigg),\nonumber
\end{align}
where $\Delta D_s = D_s-D_{s-}=-(X_s-X_{s-})=-\Delta X_s$, $\Delta U_s = U_s-U_{s-}$.
\begin{enumerate}[i)]
\item
Then, the process $p(t,D_t, U_t)$ is a supermartingale if
\subitem $I(p_C)(t,D_t, U_t)$ is a martingale,
\subitem $I(p_S)(t,D_t, U_t)$ is a supermartingale
\subitem and if the following {\it smooth fit} condition holds:
\begin{equation}\label{smooth}\frac{\partial}{\partial d}p_S(t,d,u)\big|_{d\nearrow\theta^*}
=\frac{\partial}{\partial d}p_C(t,d,u)\big|_{d\searrow\theta^*},\qquad\textrm{(\textit{smooth fit})}\end{equation}
\subitem for $\sigma >0$ and if the following {\it continuous fit} condition  holds:
\begin{equation}\label{smooth2}p_S(t,\theta^*,u)=p_C(t,\theta^*,u)\qquad\textrm{(\textit{continuous fit})}\end{equation}
\subitem for $\sigma=0$ and $X_t$ being a L\'evy process of bounded variation.
\item
If the stopped process $p_C(t\wedge \tau_D^-(\theta^*),D_{t\wedge \tau_D^-(\theta^*)}, U_{t\wedge \tau_D^-(\theta^*)})$ is a martingale then the process $I(p_C)(t,D_t, U_t)$ is also a martingale and we have
\begin{alignat}{2}
&\frac{\partial}{\partial t}p_C(t,d,u)+\mathcal{A}^{(D,U)}p_C(t,d,u)=0,\,&&\qquad\textrm{(\textit{martingale condition})}\label{mtg}\\
&\frac{\partial}{\partial d}p_C(t,d,u)|_{d=0}=0,\ \frac{\partial}{\partial u}p_C(t,d,u)|_{u=0}= 0,&&\qquad\textrm{(\textit{normal reflection})}\label{normalref}
\end{alignat}
where $\mathcal{A}^{(D,U)}$ is a full generator of
the Markov process $(D_t,U_t)$ defined as follows:
\begin{align}
&\mathcal{A}^{(D,U)}\hat{p}(t,d,u):=\nonumber\\&
\qquad-\mu \frac{\partial}{\partial d}\hat{p}(t,d,u)+\mu \frac{\partial}{\partial u}\hat{p}(t,d,u)+\frac{\sigma^2}{2}\frac{\partial^2}{\partial d^2}\hat{p}(t,d,u)+\frac{\sigma^2}{2}\frac{\partial^2}{\partial u^2}\hat{p}(t,d,u)-\frac{\sigma^2}{2}\frac{\partial^2}{\partial u\partial d}\hat{p}(t,d,u)\nonumber\\&
\qquad +\int_{(0,\infty)}\left(\hat{p}(t,d+z,(u-z)\vee 0)-\hat{p}(t,d,u)-z\frac{\partial}{\partial d}\hat{p}(t,d,u)\mathbbm{1}_{(z<1)}
+z\frac{\partial}{\partial u}\hat{p}(t,d,u)\mathbbm{1}_{(z<1)}\right)\Pi(-\diff z),\label{Dgenerator}
\end{align}
for $\sigma>0$ and
\begin{align}
\mathcal{A}^{(D,U)}\hat{p}(t,d,u):=-\hat{\mu} \frac{\partial}{\partial d}\hat{p}(t,d,u)+\hat{\mu} \frac{\partial}{\partial u}\hat{p}(t,d,u)+\int_{(0,\infty)}\left(\hat{p}(t,d+z,(u-z)\vee 0)-\hat{p}(t,d)\right)\Pi(-\diff z),\label{generator_bounded}
\end{align}
for $\sigma=0$ and $X_t$ being a L\'evy process with bounded variation.
\item The process $I(p_S)(t,D_t, U_t)$ is a supermartingale if:
\begin{alignat}{2}
&\frac{\partial}{\partial t}p_S(t,d,u)+\mathcal{A}^{(D,U)}p_S(t,d,u)\leq 0,&&\qquad\textrm{(\textit{supermartingale condition})}\label{supmtg}\\
&\frac{\partial}{\partial d}p_S(t,d,u)|_{d=0}\geq 0,\ \frac{\partial}{\partial u}p_S(t,d,u)|_{u=0}\geq 0.\label{normalref2}
\end{alignat}
\end{enumerate}
\end{lem}
\begin{proof}
See Appendix.
\end{proof}
\begin{Rem}\label{rem_lem7}
This lemma considers the case of function $p(t,D_t,U_t)$, which may depend on $t$, $D_t$ and $U_t$. In the framework of the drawdown contract we take
\begin{align}
p(t,d,u)=p(t,d)=e^{-rt}g_{\tau_D^-(\theta^*)}(d,p)\nonumber
\end{align}
though, which depends on $t$ and $d$ only. In this case, this lemma still holds true by taking all derivatives with respect to $u$ equal to $0$. Thus, we simplify the generator $\mathcal{A}^{(D,U)}$, given in \eqref{Dgenerator}, to the full generator of Markov process $D_t$ only:
\begin{align}
\mathcal{A}^{(D)}\hat{p}(t,d):=-\mu \frac{\partial}{\partial d}\hat{p}(t,d)+\frac{\sigma^2}{2}\frac{\partial^2}{\partial d^2}\hat{p}(t,d)+\int_{(0,\infty)}\left(\hat{p}(t,d+z)-\hat{p}(t,d)-z\frac{\partial}{\partial d}\hat{p}(t,d)\mathbbm{1}_{(z<1)}\right)\Pi(-\diff z),\nonumber
\end{align}
for $\sigma>0$ and
\begin{align}
\mathcal{A}^{(D)}\hat{p}(t,d):=-\hat{\mu} \frac{\partial}{\partial d}\hat{p}(t,d)+\int_{(0,\infty)}\left(\hat{p}(t,d+z)-\hat{p}(t,d)\right)\Pi(-\diff z),\nonumber
\end{align}
for $\sigma=0$ and $X_t$ being a L\'evy process with bounded variation.
\end{Rem}

The main message of this crucial lemma is that we can separately analyse the continuation and stopping regions.
In other words, to check the supermartingale condition of Verification Lemma \ref{war},
it is enough to prove the \textit{martingale property} in so-called continuation region $C$ and the \textit{supermartingale property} in so-called stopping region
$S$. This is possible thanks to the smooth/continuous fit condition that holds at the boundary that links these two regions.
We can now solve our optimisation problem.

\begin{theorem} Assume that \eqref{mainzalozenia} holds.
Let non-increasing bounded penalty function $c: \mathbb{R}_+\rightarrow\mathbb{R}_+$ be in $C^2([0,a))$ and satisfies:
\begin{align}
-rc(d)-\mu c'(d)+\frac{\sigma^2}{2}c''(d)+\int_{(0,\infty)}\left(c(d+z)-c(d)-zc'(d)\mathbbm{1}_{(z<1)}\right)\Pi(-\diff z) \geq -p.\label{war1}
\end{align}
If
\begin{align}
\int_{(\theta^*-d,\infty)}\tilde{f}(d+z,p)\Pi(-\diff z)\geq 0\ \textrm{for all }d\in[0,\theta^*),\label{assum_additional}
\end{align}
then the stopping time $\tau_D^-(\theta^*)$ for $\theta^*$
defined by (\ref{teta*}) is the optimal stopping rule for the stopping problem (\ref{G}) and hence also for the insurance contract \eqref{F}.
\end{theorem}
\begin{proof} We have to show that $\tau_D^-(\theta^*)$
and, hence, $g_{\tau_D^-(\theta^*)}(d,p)$ satisfy all of the conditions of the Verification Lemma \ref{war}.
We take Markov process $\Upsilon_t=D_t$, $\mathbb{B}=\mathbb{R}_+$, $\tau_0=\tau_D^+(a)$ and
$V(\phi)=g_{\tau_D^-(\theta^*)}(d,p)$ with $\phi=d$.
To prove the domination condition (i) of Verification Lemma \ref{war},
we have to show that $g_{\tau_D^-(\theta^*)}(d,p)\geq \tilde{f}(d,p)$.
To prove this, note that
from the definition of $\theta^*$ we have that $g_>(d,p,\theta^*)\geq g_>(d,p,\theta)$ for all $\theta\in(0,a)$.
Furthermore, for $d>\theta^*$ we have:
\begin{align}
g_{\tau_D^-(\theta^*)}(d,p)=g_>(d,p,\theta^*)\geq g_>(d,p,d)=\tilde{f}(d,p).\nonumber
\end{align}
Similarly, for $d\leq\theta^*$,
\begin{align}
g_{\tau_D^-(\theta^*)}(d,p)=g_<(d,p,\theta^*)=\tilde{f}(d,p)\nonumber
\end{align}
which completes the proof of condition i) of Verification Lemma \ref{war}.

To prove the second condition (ii) of Verification Lemma \ref{war}, we have to show that the process $$\left\{e^{-r(t\wedge\tau_D^+(a))}g_{\tau_D^-(\theta^*)}(D_{t\wedge\tau_D^+(a)},p)\right\}$$ is a supermartingale.
To do this, we prove that $\left\{e^{-rt}g_{\tau_D^-(\theta^*)}(D_{t},p)\right\}$ is a supermartingale by
applying Lemma \ref{supermtg} (i)
for function $p(t,d,u):=p(t,d)=e^{-rt}g_{\tau_D^-(\theta^*)}(d,p)$, for $d\geq 0$ and $C=(\theta^*,\infty)$, $S=[0,\theta^*)$.
Given that $\theta^*$ satisfies the \textit{continuous fit} condition \eqref{continuous_fit} we can
separately consider functions $p_C(t,d):=e^{-rt}g_>(d,p,\theta^*)\mathbbm{1}_{(d\in (\theta^*,\infty))}$ and $p_S(t,d):=e^{-rt}g_<(d,p,\theta^*)\mathbbm{1}_{(d\in [0,\theta^*))}$.
Moreover, the {\it smooth fit} condition \eqref{smooth_fit} is also fulfilled. Thus, according to Lemma \ref{supermtg}(i), it is enough to show that $I(p_C)(t,D_t)$ is a martingale and $I(p_S)(t,D_t)$ is a supermartingale.

At the beginning note that the function $g_{\tau_D^-(\theta^*)}(d,p)$ defined in \eqref{gtaub} is bounded.
Indeed, this follows from the fact that $\tilde{f}$ is bounded
which is a consequence of the inequality $\xi(\cdot)\leq 1$, assumption \eqref{warXi} and assumed conditions on fee function $c(\cdot)$.
Thus, by Lemma \ref{supermtg}(ii), the \textit{martingale} condition of $I(p_C)(t,D_t)$ follows form a martingale property of the process:
$$\left\{ e^{-r(t\wedge\tau_D^-(\theta^*))}g_>(D_{t\wedge\tau_D^-(\theta^*)},p,\theta^*)\right\}.$$
This property follows from the Strong Markov Property:
\begin{align}
&\E\left[e^{-r\tau_D^-(\theta^*)}g_>(D_{\tau_D^-(\theta^*)},p,\theta^*)\big|\ \mathcal{F}_{t\wedge\tau_D^-(\theta^*)}\right]\nonumber\\
&\qquad = \E\left[e^{-r\tau_D^-(\theta^*))}\E_{|D_{\tau_D^-(\theta^*)}}\left[e^{-r\tau_D^-(\theta^*)}\tilde{f}(\theta^*,p)\right]\big|\ \mathcal{F}_{t\wedge\tau_D^-(\theta^*)}\right]\nonumber\\
&\qquad = e^{-r(t\wedge\tau_D^-(\theta^*))}g_>(D_{t\wedge\tau_D^-(\theta^*)}).\nonumber
\end{align}

Now, by Lemma \ref{supermtg}, we finish the proof if we show that the \textit{supermartingale} condition for $e^{-rt}g_<(d,p,\theta^*)\mathbbm{1}_{(d\in [0,\theta^*))}$ is satisfied.
To do so, we first prove that $e^{-rt}\tilde{f}(D_t,p)$ is a supermartingale.
Using definition of $\tilde{f}$ given in \eqref{tildef}, we can write:
\begin{align}
e^{-rt}\tilde{f}(d,p)=-\frac{p}{r}e^{-rt}\xi(d)-e^{-rt}\Xi(d)+e^{-rt}\left(\frac{p}{r}-c(d)\right).
\end{align}
Moreover, the processes $e^{-rt}\xi(D_t)$ and $e^{-rt}\Xi(D_t)$ are both  martingales by Strong Markov Property.
Thus, for a supermartingale property for $e^{-rt}\tilde{f}(D_t,p)$ holds if a process $e^{-rt}\left(\frac{p}{r}-c(d)\right)$ is a supermartingale.
Using \cite[Thm. 31.5]{sato},
the latter statement is equivalent to requirement that
\begin{align}
\mathcal{A}^{(D)}\left(\frac{p}{r}-c(d)\right)-r\left(\frac{p}{r}-c(d)\right)&\leq 0\label{adu}
\end{align}
which holds true by assumption \eqref{war1} and that
\begin{align}
\frac{\partial}{\partial d}\left(\frac{p}{r}-c(d)\right)|_{d=0}=-c'(0)\geq 0\label{normalrefld}
\end{align}
which follows from the fact that the function $c$ is non-increasing.

To prove \textit{supermartingale} property of $I(p_S)$ for $p_S(t,d)=e^{-rt}g_<(d,p,\theta^*)\mathbbm{1}_{(d\in [0,\theta^*))}$ we apply Lemma \ref{supermtg}(iii). That is, this property holds if for $d\in[0,\theta^*)$ we have:
\begin{align}
&-re^{-rt}g_<(d,p,\theta^*)+e^{-rt}\mathcal{A}^{(D)}g_<(d,p,\theta^*)=-re^{-rt}g_<(d,p,\theta^*)-\mu e^{-rt}\frac{\partial}{\partial d}g_<(d,p,\theta^*)+\frac{\sigma^2}{2}\frac{\partial^2}{\partial d^2}g_<(d,p,\theta^*)\nonumber\\
&\qquad\qquad\qquad+e^{-rt}\int_{(0,\infty)}\left(g_<(d+z,p,\theta^*)\mathbbm{1}_{(d+z\in[0,\theta^*))}-g_<(d,p,\theta^*)-z\frac{\partial}{\partial}g_<(d,p,\theta^*)\right)\Pi(-\diff z)\leq 0.\nonumber
\end{align}
Note that, from \eqref{g=f}, we have that $g_<(d,p,\theta^*)=\tilde{f}(d,p)$ for $d\in[0,\theta^*)$.
Moreover, the indicator appearing in the definition of $p_S$  is important only in first increment in the sum under above integral.
More precisely, because the process $D$ only jumps  upward, the above mentioned indicator may produce zero after a possible jump.
Taking this observation into account, we can rewrite the above inequality as follows:
\begin{align}
&-re^{-rt}g_<(d,p,\theta^*)+e^{-rt}\mathcal{A}^{(D)}g_<(d,p,\theta^*)=-re^{-rt}\tilde{f}(d,p)-\mu e^{-rt}\frac{\partial}{\partial d}\tilde{f}(d,p)+\frac{\sigma^2}{2}\frac{\partial^2}{\partial d^2}\tilde{f}(d,p)\nonumber\\
&\quad+e^{-rt}\int_{(0,\theta^*-d)}\left(\tilde{f}(d+z,p)-\tilde{f}(d,p)-z\frac{\partial}{\partial}\tilde{f}(d,p)\right)\Pi(-\diff z)+e^{-rt}\int_{(\theta^*-d,\infty)}\left(-\tilde{f}(d,p)-z\frac{\partial}{\partial d}\tilde{f}(d,p)\right)\Pi(-\diff z)\nonumber\\
&\leq-re^{-rt}\tilde{f}(d,p)-\mu e^{-rt}\frac{\partial}{\partial d}\tilde{f}(d,p)+\frac{\sigma^2}{2}\frac{\partial^2}{\partial d^2}\tilde{f}(d,p)+e^{-rt}\int_{(0,\infty)}\left(\tilde{f}(d+z,p)-\tilde{f}(d,p)-z\frac{\partial}{\partial}\tilde{f}(d,p)\right)\Pi(-\diff z)\leq 0.\nonumber
\end{align}
The first inequality follows from \eqref{assum_additional} and the second follows from supermartingale property of $\tilde{f}$, as proven previously.

The condition \eqref{normalref2} in Lemma \ref{supermtg} (iii) is satisfied since $g_<(d,p,\theta^*)=\tilde{f}(d,p)$ on $d\in[0,\theta^*)$ and condition \eqref{normalref2} was already proved for $\tilde{f}$. This completes the proof.
\end{proof}
\begin{Rem}\label{rem1}
We will give now a few comments about made assumptions \eqref{war1} and \eqref{assum_additional}.

Note that condition \eqref{war1} is satisfied when $-\mu+\int_{(1,\infty)}z\Pi(-\diff z)\leq 0$ and for all convex non-increasing penalty functions $c$ such that
\begin{align}
\frac{p}{r}\geq c(d) \qquad \textrm{for all }d\geq 0.\label{condition_for_convex}
\end{align}
Indeed, we can rewrite inequality \eqref{war1} as follows
\begin{align}
&-\mu c^\prime(d)+\frac{\sigma^2}{2}c^{\prime\prime}(d)+\int_{(0,\infty)}\left(c(d+z)-c(d)-zc^\prime(d)\mathbbm{1}_{(z<1)}\right)\Pi(-\diff z) \nonumber\\
=&\left(-\mu+\int_{(1,\infty)}z\Pi(-\diff z)\right)c^\prime(d)+\frac{\sigma^2}{2}c^{\prime\prime}(d)+\int_{(0,\infty)}\left(c(d+z)-c(d)-zc^\prime(d)\right)\Pi(-\diff z)\geq -p+rc(d).\label{dodatkowe}
\end{align}
The expression under the integral sign is positive since
\begin{align}
c(d+z)-c(d)\geq zc^\prime (d)\nonumber
\end{align}
by convexity of $c$.
Now, the right-hand side of inequality \eqref{dodatkowe}
is non-positive by \eqref{condition_for_convex} and all the terms on the left-hand side are non-negative.
Thus, the condition \eqref{war1} holds true.

Condition \eqref{assum_additional} holds for the Brownian motion because it has continuous trajectories.
Unfortunately, we are unable to give any sufficient conditions for the assumption \eqref{assum_additional} to hold true.
We check this numerically  and we show that it holds for all our examples.
Note that this condition is simply satisfied whenever the payoff function is non-negative (e.g. for classical American options).
In this paper, this is not the case. Indeed, in view of decomposition in Proposition \ref{FDecomp}, the payoff function $\tilde{f}$ can attain negative values.
\end{Rem}

\addtocounter{Ex}{-2}
\begin{Ex}[continued]
Let us consider two penalty functions: linear and quadratic:
\begin{align}
c_1(d)&=\frac{p}{ra}\left(a-d\right)\mathbbm{1}_{(0\leq d<a)},\label{c1_ex}\\
c_2(d)&=\frac{p}{ra^2}\left(a-d\right)^2\mathbbm{1}_{(0\leq d<a)}.\label{c2_ex}
\end{align}
We choose fixed reward function; that is. $\alpha(a)=\alpha$.
We also have to choose the premium intensity $p$ such that the condition \eqref{mainzalozenia} is satisfied.
Moreover, note that for the Brownian motion the condition \eqref{assum_additional} always holds true because of absence of jumps.

\begin{figure}[!ht]
\centering
\includegraphics[width=0.5\textwidth]{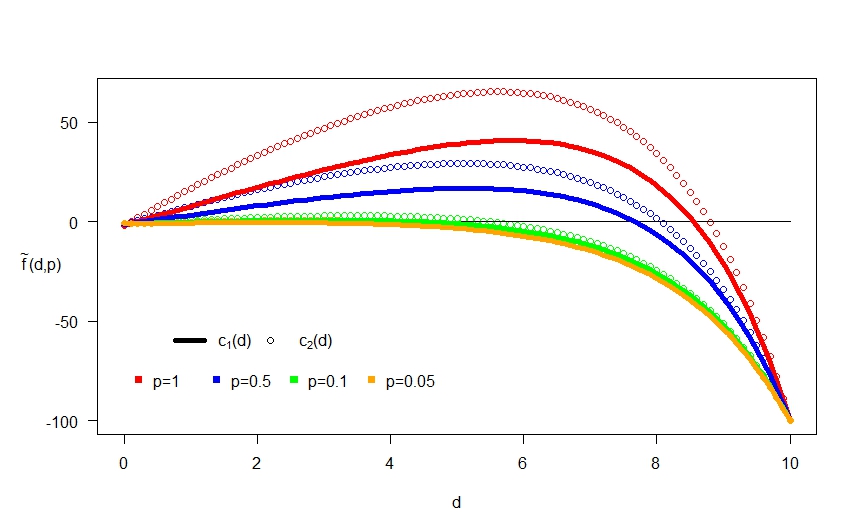}
\caption{\footnotesize{The function $\tilde{f}$ for Brownian motion, linear penalty function $c_1$ and quadratic penalty function $c_2$, and for various premiums $p$. Parameters: $r=0.01, \mu=0.03, \sigma =0.4, a=10, \alpha=100$.}}
\label{fig:ftildebm}
\end{figure}
Given that the function $\tilde{f}$ is decreasing with respect to premium $p$ for both $c_1(\cdot)$ and $c_2(\cdot)$, Figure \ref{fig:ftildebm} shows that
we can choose $p\geq 0.1$ in both cases for our set of parameters. Recall that the optimal $\theta^*$ maximises function $g_>(d,p,\theta)$ given in \eqref{g}. For Brownian motion we have:
\begin{align}
g_>(d,p,\theta)=e^{\frac{\mu}{\sigma^2}(d-\theta )}\frac{\sinh (\Xi (a-d))}{\sinh (\Xi (a-\theta))}\left(-\left(\frac{p}{r}+\alpha\right)e^{-\frac{\mu}{\sigma^2}(a-\theta )}\frac{\Xi\cosh(\Xi \theta )-\frac{\mu}{\sigma^2}\sinh(\Xi \theta )}{\Xi\cosh(\Xi a)-\frac{\mu}{\sigma^2}\sinh(\Xi a)}+\frac{p}{r}-c(\theta)\right).\nonumber
\end{align}
By making a plot of this function, we can easily identify the optimal level $\theta^*$. This is done in
Figure \ref{fig:gbm}.

\begin{figure}[!ht]
\centering
\includegraphics[width=0.5\textwidth]{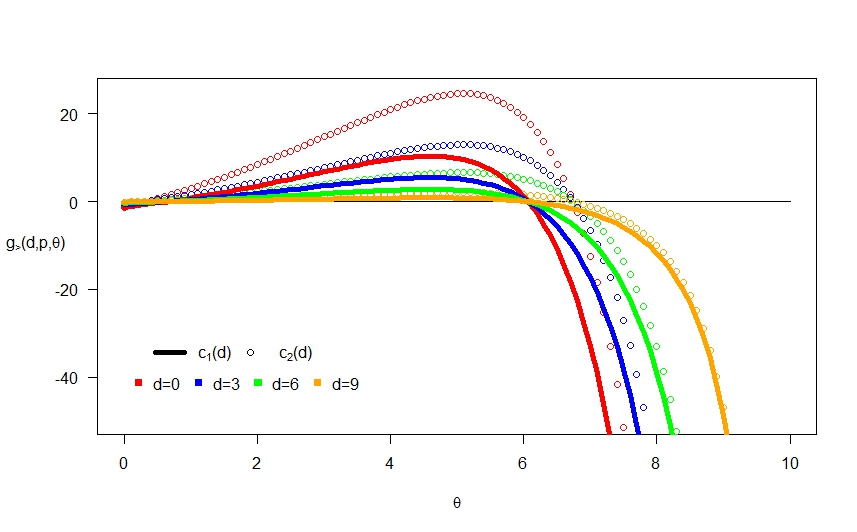}
\caption{\footnotesize{The function $g_>(d,p,\theta )$ for Brownian motion, linear penalty function $c_1$ and quadratic penalty function $c_2$, and for various staring drawdown levels $d$. Parameters: $p=0.2, r=0.01, \mu=0.03, \sigma =0.4, a=10, \alpha=100$.}}
\label{fig:gbm}
\end{figure}
\end{Ex}

\begin{Ex}[continued]
We also analyse Cram\'er-Lundberg model for the linear and quadratic penalty functions \eqref{c1_ex}-\eqref{c2_ex} and for the linear reward function $\alpha (\cdot)$.
As Figure \ref{fig:ftildecl} of the function $\tilde{f}$ shows, we can take a penalty intensity $p$ bigger than $0.1$ for our set of parameters for condition
\eqref{mainzalozenia} to hold.
Moreover, in this case, we have:
\begin{align}
g_>(d,p,\theta)=&\left(\frac{p}{r}-c(\theta )-\left(\frac{p}{r}+\alpha(a+\frac{1}{\rho})\right)\left(c_0 + c_{\Phi(r)}e^{\Phi(r) (a-\theta)}+c_{\zeta}e^{\zeta (a-\theta)}\right)\right)\nonumber\\
&\quad\cdot\frac{\psi^\prime (\zeta )e^{\Phi (r)(a-d)}+\psi^\prime (\Phi (r))e^{\zeta (a-d)}}{\psi^\prime (\zeta )e^{\Phi (r)(a-\theta)}+\psi^\prime (\Phi (r))e^{\zeta (a-\theta)}}.\nonumber
\end{align}
In Figure \ref{fig:gcl}, we numerically check the condition \eqref{assum_additional} and present the value of $g_>(d,p,\theta)$ using the above formula. To find optimal stopping level $\theta^*$ for the drawdown, we pick the $\theta$ that maximises the value function.

\begin{figure}[!ht]
\centering
\includegraphics[width=0.5\textwidth]{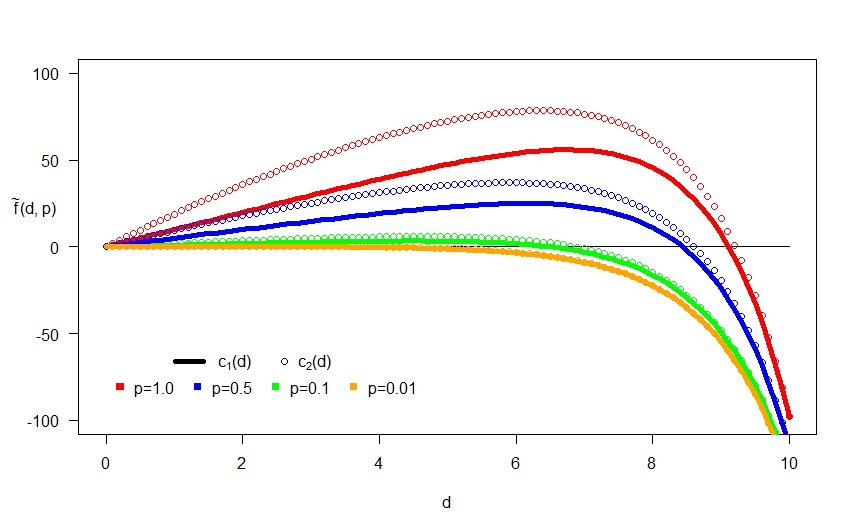}
\caption{\footnotesize{The function $\tilde{f}$ for Cram\'er-Lundberg, linear $c_1$ and quadratic $c_2$ penalty functions, and for various premiums $p$. Parameters: $r=0.01, \hat{\mu}=0.05, \beta =0.1, \rho=2.5, a=10, \alpha (d)=100+10d$.}}
\label{fig:ftildecl}
\end{figure}
\begin{figure}[!ht]
\center{
\subfloat{\includegraphics[width=0.5\textwidth]{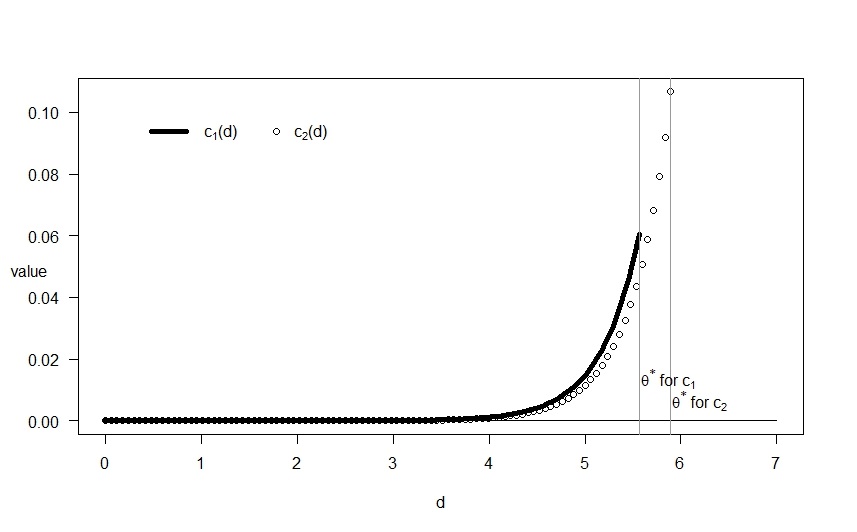}}
\subfloat{\includegraphics[width=0.5\textwidth]{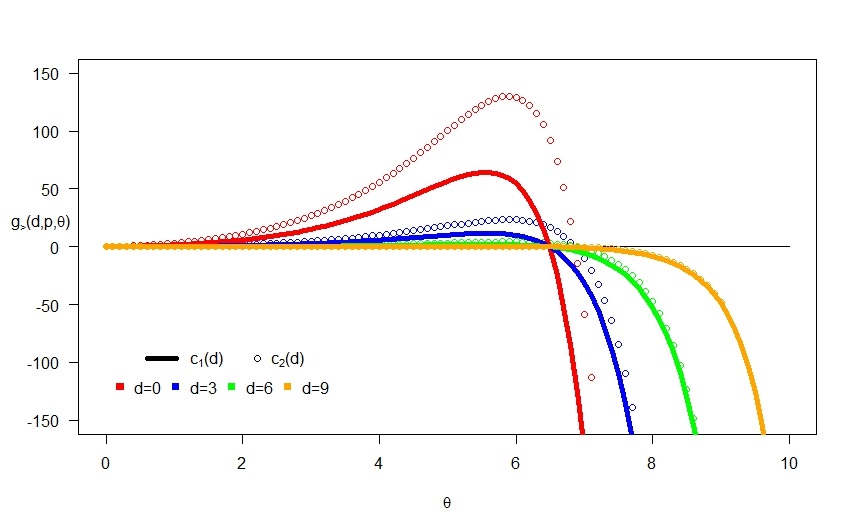}}}
\caption{\footnotesize{The condition \eqref{assum_additional} (left) and the function $g_>(d,p,\theta )$ (right) for Cram\'er-Lundberg model, linear $c_1$ and quadratic $c_2$ penalty functions, and for various starting drawdown levels $d$. Parameters: $p=0.1, r=0.01, \hat{\mu}=0.05, \beta =0.1, \rho =2.5, a=10, \alpha (d)=100+10d$.}}
\label{fig:gcl}
\end{figure}
\end{Ex}

\section{Incorporating drawup contingency}\label{sec:drawup}
\subsection{Fair premium}
The investors might like to buy a contract that meets some maturity conditions.
This means that contract will end when these conditions will be fulfilled.
We add this feature to the previous contracts by considering drawup contingency.
In particular, the next contracts may expire earlier if a fixed drawup event occurs prior to drawdown epoch.
Choosing the drawup event is natural since it corresponds to some market upward trends.
Therefore, the investor might stop believing that a substantial drawdown will happen in the close future
and she/he might want to stop paying a premium when this event happens.
Under a risk-neutral measure, the value of this contract equals:
\begin{align}
k(d,u,p):=&\mathbb{E}_{|d|u}\left[-\int_0^{\tau^+_D (a)\wedge\tau^+_U(b)}e^{-rt}p\diff t+\alpha(D_{\tau^+_D (a)}) e^{-r\tau^+_D(a)}\mathbbm{1}_{\left(\tau^+_D(a)<\tau^+_U(b)\right)}\right],\label{kyzp}\end{align}
where $a\geq b\geq 0$.
At the beginning we find the above value function and later we identify the fair premium $p^*$ under which
\begin{equation}\label{fairk}
k(d,u,p^*)=0.
\end{equation}

Note that
\begin{align}
k(d,u,p)=&\frac{p}{r}\Big(\nu(d,u)+\lambda(d,u)\Big)+N(d,u)-\frac{p}{r},\label{fdrawup}
\end{align}
where
\begin{align}
\nu(d,u):=&\mathbb{E}_{|d|u}\left[e^{-r\tau_D^+(a)};\tau_D^+(a)<\tau_U^+(b)\right],\label{nu}\\
\lambda(d,u):=&\mathbb{E}_{|d|u}\left[e^{-r\tau_U^+(b)};\tau_U^+(b)<\tau_D^+(a)\right],\label{lambda}\\
N(d,u):=&\mathbb{E}_{|d|u}\left[e^{-r\tau_D^+(a)}\alpha(D_{\tau_D^+(a)});\tau_D^+(a)<\tau_U^+(b)\right].\label{N}
\end{align}
Note that all of the above functions are well defined since $\nu,\lambda\in[0,1]$ and $N(d,\cdot)\leq\Xi(d)<\infty$ by \eqref{warXi}.

In the next proposition, we identify all of the above quantities.
We denote $x^+=\max(x,0)$.
\begin{Prop}\label{lambdanu}
Let $b<a$. For any $d\in[0,a]$ and $u\in[0,b]$ we have:
\begin{align}
\lambda(d,u)=&\frac{W^{(r)}((a-d)\wedge u)}{W^{(r)}(b-u+(a-d)\wedge u)}+\mathbb{E}\left[e^{-r\tau_U^+(b)};\overline{X}_{\tau_U^+(b)}\vee d-\underline{X}_{\tau_U^+(b)}<a,\underline{X}_{\tau_U^+(b)}\leq -u\right],\nonumber\\
\nu(d,u)=&\left(Z^{(r)}(a-d)-Z^{(r)}(b-u+a-d)\frac{W^{(r)}(a-d)}{W^{(r)}(b-u+a-d)}\right)\mathbbm{1}_{(d+u\geq a)}\nonumber\\
&\quad+\mathbb{E}\left[e^{-r\tau_U^+(b)};\overline{X}_{\tau_U^+(b)}\vee d-\underline{X}_{\tau_U^+(b)}\geq a,\ \underline{X}_{\tau_U^+(b)}\leq -u,\ \overline{X}_{\tau_U^+(b)}\leq b-u\right]\frac{\mathbbm{1}_{(d+u<a)}}{\E\left[e^{-r\tau_U^+(b)}\right]},\nonumber\\
N(d,u)=&\Xi(d)-\frac{W^{(r)}((a-d)\wedge u)}{W^{(r)}(b-u+(a-d)\wedge u)}\Xi((d+u-b)^+)\nonumber\\
&-\int_{((d+u-b)^+,\ a-b)}\Xi(s)\mathbb{E}_{|d|u}\left[e^{-r\tau_U^+(b)};\overline{X}_{\tau_U^+(b)}\vee d-\underline{X}_{\tau_U^+(b)}\in\diff s+b,\ \underline{X}_{\tau_U^+(b)}\leq -u\right],\nonumber
\end{align}
where the joint distribution for $(\tau_U^+(b),\overline{X}_{\tau_U^+(b)},\underline{X}_{\tau_U^+(b)})$ is given via \eqref{m2}
and where $\Xi(\cdot)$ is defined in \eqref{Xi}.\\
Moreover, if $a=b$ then, for $d,u\in [0,a]$, we have:
\begin{align}
\lambda (d,u) =& \frac{W^{(r)}(a-d)}{W^{(r)}(a-u+(a-d)\wedge u)}-\frac{1}{r}\frac{W^{\prime (r)}(a)}{(W^{(r)}(a))^2}\left(Z^{(r)}((a-d)\vee u)-Z^{(r)}(u)\right),\nonumber\\
\nu (d,u) =& Z^{(r)}((a-d)\wedge u)-Z^{(r)}(a-u+(a-d)\wedge u)\frac{W^{(r)}(a-d)}{W^{(r)}(a-u+(a-d)\wedge u)}\nonumber\\
&+\frac{1}{r}Z^{(r)}(a)\frac{W^{\prime (r)}(a)}{(W^{(r)}(a))^2}\left(Z^{(r)}((a-d)\vee u)-Z^{(r)}(u)\right),\nonumber\\
N(d,u) =& \Xi(d)-\lambda(d,u)\Xi((d+u-a)^+).\nonumber
\end{align}
\end{Prop}
\begin{proof}
See Appendix.
\end{proof}

Identity \eqref{fdrawup} gives the following theorem.
\begin{theorem}\label{thm2}
The price of the contract \eqref{kyzp} is given in \eqref{fdrawup} and
the fair premium defined in \eqref{fairk} equals:
\begin{equation}
p^*(d,u) = \frac{rN(d,u)}{1-\lambda(d,u)-\nu(d,u)},\label{p*drawup}
\end{equation}
where functions $\lambda$, $\nu$ and $N$ are given in Proposition \ref{lambdanu}.
\end{theorem}

\addtocounter{Ex}{-2}
\begin{Ex}[continued]
We continue analysing the case of linear Brownian motion defined in \eqref{linearBrownian}.
Assume that $b<a$.
To find the price of the contract \eqref{kyzp}, we use formula \eqref{fdrawup}
and we calculate all of the functions $\lambda$, $\nu$ and $N$ given in \eqref{nu}-\eqref{N}.

At the beginning, let us consider a case of $a\leq d+u$.
By Lemma \ref{lambdanu}, the expressions for functions $\nu$ and $\lambda$ reduce to two-sided exit formulas
\eqref{twosided1}--\eqref{twosided2}. These are given explicitly in terms of the scale functions:
\begin{align}
\lambda(d,u)&=\frac{W^{(r)}(a-d)}{W^{(r)}(a+b-d-u)},\nonumber\\
\nu(d,u)&=Z^{(r)}(a-d)-Z^{(r)}(a+b-d-u)\frac{W^{(r)}(a-d)}{W^{(r)}(a+b-d-u)}.\nonumber
\end{align}
Given that Brownian motion has continuous trajectories,
we have $D_{\tau_D^+(a)}=a$. Denoting $\alpha (a)=\alpha$
from the definition \eqref{N} of the function $N$,  we can conclude that $$N(d,u)=\alpha\nu(d,u).$$

The case of $a>d+u$ is slightly more complex.  First observe that the dual process $\widehat{X}_t:=-X_t$
is again a linear Brownian motion with drift $-\mu$.
Moreover, the geometrical observation of the trajectories of the processes $X$ and $\widehat{X}$
gives that $\widehat{U}_t=D_t$ and $\widehat{D}_t=U_t$, where $\widehat{U}_t$ and $\widehat{D}_t$ are the drawup and drawdown processes for $\widehat{X}_t$, respectively.
Now, using \cite{Mijatovic&Pistorius}, we can find explicit expressions for the functions $\lambda$ and $\nu$:
\begin{align}
&\lambda(d,u)=\frac{1}{r}\frac{1}{\WW^{(r)}(b)}\frac{\sigma^2}{2}\left[\frac{(\WW^{\prime (r)}(b))^2}{\WW^{(r)}(b)}-\WW^{\prime\prime (r)}(b)\right]\left(\ZZ^{(r)}(b\wedge(a-u))-\ZZ^{(r)}(d)\right)e^{-(a-b)\frac{\WW^{\prime (r)}(b)}{\WW^{(r)}(b)}}\mathbbm{1}_{(b>d)}\nonumber\\
&\qquad+\frac{\WW^{(r)}(b)}{\WW^{\prime (r)}(b)}\frac{\sigma^2}{2}\left[\frac{(\WW^{\prime (r)}(b))^2}{\WW^{(r)}(b)}-\WW^{\prime\prime (r)}(b)\right]\left(e^{-u\frac{\WW^{\prime (r)}(b)}{\WW^{(r)}(b)}}-e^{-(a-b\vee d)\frac{\WW^{\prime (r)}(b)}{\WW^{(r)}(b)}}\right)\mathbbm{1}_{(a>u+b)}+\frac{W^{(r)}(u)}{W^{(r)}(b)}\nonumber\\
&\nu(d,u)=\frac{\WW^{(r)}(b-u)}{\WW^{\prime (r)}(b)}\frac{\sigma^2}{2}\left[\frac{(\WW^{\prime (r)}(b))^2}{\WW^{(r)}(b)}-\WW^{\prime\prime (r)}(b)\right]e^{-(a-b\vee (d+u))\frac{\WW^{\prime (r)}(b)}{\WW^{(r)}(b)}}Z^{(r)}(b)\nonumber\\
&\qquad -\frac{1}{r}\frac{1}{\WW^{(r)}(b)}\frac{\sigma^2}{2}\left[\frac{(\WW^{\prime (r)}(b))^2}{\WW^{(r)}(b)}-\WW^{\prime\prime (r)}(b)\right]\left(\ZZ^{(r)}(b-u)-\ZZ^{(r)}(d)\right)e^{-(a-b)\frac{\WW^{\prime (r)}(b)}{\WW^{(r)}(b)}}Z^{(r)}(b)\mathbbm{1}_{(b>d+u)},\nonumber
\end{align}
where $\WW^{(r)}$ and $\ZZ^{(r)}$ are the scale functions defined for $\widehat{X}_t$ (see also \cite{drawdownup1} and \cite{ZP&JT} for all details of calculations).

Because $N(d,u)=\alpha\nu(d,u)$,  the contract value and the fair premium defined in \eqref{fdrawup} and \eqref{p*drawup} can be represented as follows:
\begin{align}
&k(d,u,p)=\left(\frac{p}{r}+\alpha\right)\nu(d,u)+\frac{p}{r}\lambda(d,u)-\frac{p}{r},\nonumber\\
&p^*(d,u)=\frac{r\alpha\nu(d,u)}{1-\lambda(d,u)-\nu(d,u)}.\nonumber
\end{align}
Figures \ref{fig:k_bm} and \ref{fig:p*_k_bm} show the contract value and the fair premium levels depending on the starting positions of the drawdown and drawup.

\begin{figure}[!ht]
\centering
\includegraphics[width=0.60\textwidth]{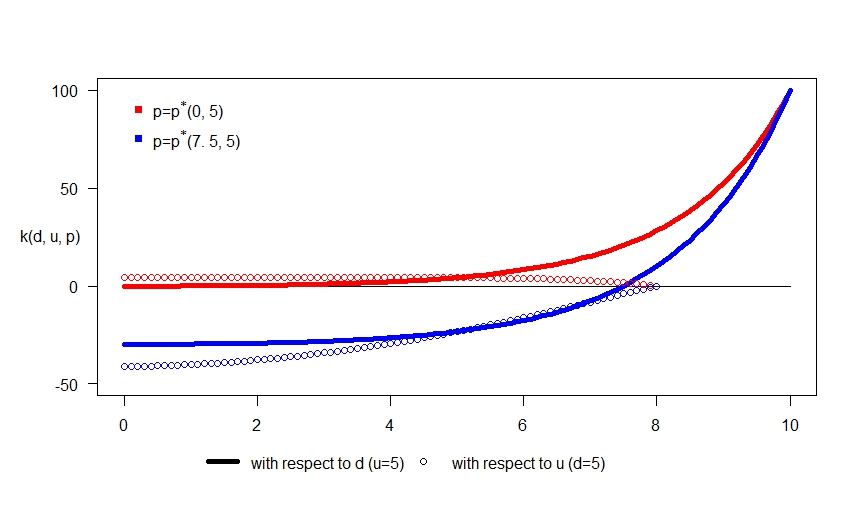}
\caption{\footnotesize{The value of the contract $k(d,u,p)$ with drawup contingency for linear Brownian motion
depending on the starting positions of drawdown and drawup for various premium levels $p$. Parameters: $r=0.01,\ \mu=0.03,\ \sigma=0.4,\ a=10,\ b=8,\ \alpha=100.$}}
\label{fig:k_bm}
\end{figure}

\begin{figure}[!ht]
\center{
\subfloat{\includegraphics[width=0.60\textwidth]{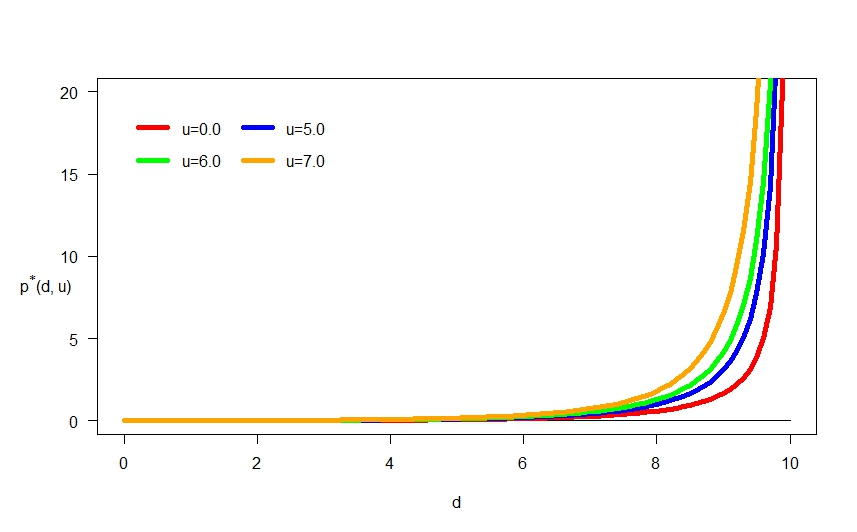}}
\quad
\subfloat{\includegraphics[width=0.60\textwidth]{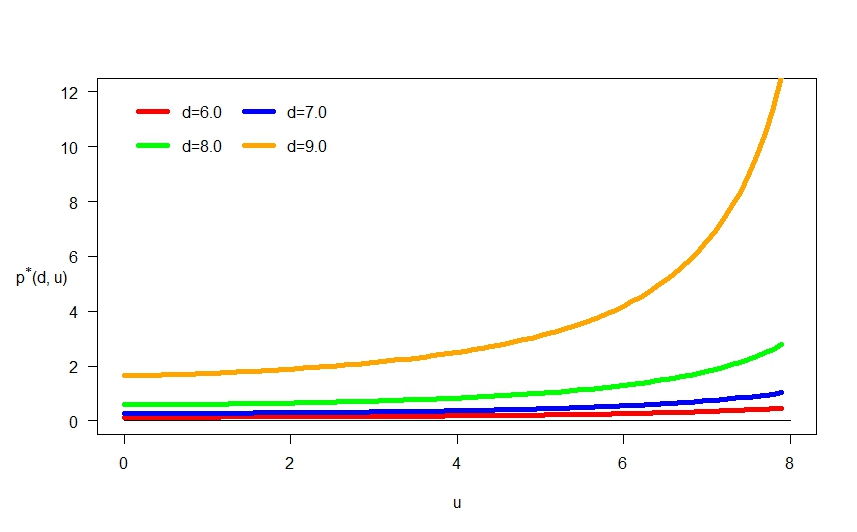}}
\caption{\footnotesize{The fair premium $p^*$ for the contract
$k(d,u,p)$ with drawup contingency for linear Brownian motion depending on the starting positions of drawdown (top) and drawup (bottom). Parameters: $r=0.01,\ \mu=0.03,\ \sigma=0.4,\ a=10,\ b=8,\ \alpha=100.$}}\label{fig:p*_k_bm}}
\end{figure}
\end{Ex}

\begin{Ex}[continued]
Now, assume that $a=b$ and that the reward function is linear. To analyse the
Cram\'er-Lundberg process we use Proposition \ref{lambdanu}, which identifies the contract value given in \eqref{fdrawup}.
The fair premium $p^*$ can be derived from \eqref{p*drawup} of Theorem \ref{thm2}.

Figures \ref{fig:k_cl} and \ref{fig:p*_k_cl} show the contract value and the fair premium levels depending on the starting positions of the drawdown and drawup.

\begin{figure}[!ht]
\centering
\includegraphics[width=0.60\textwidth]{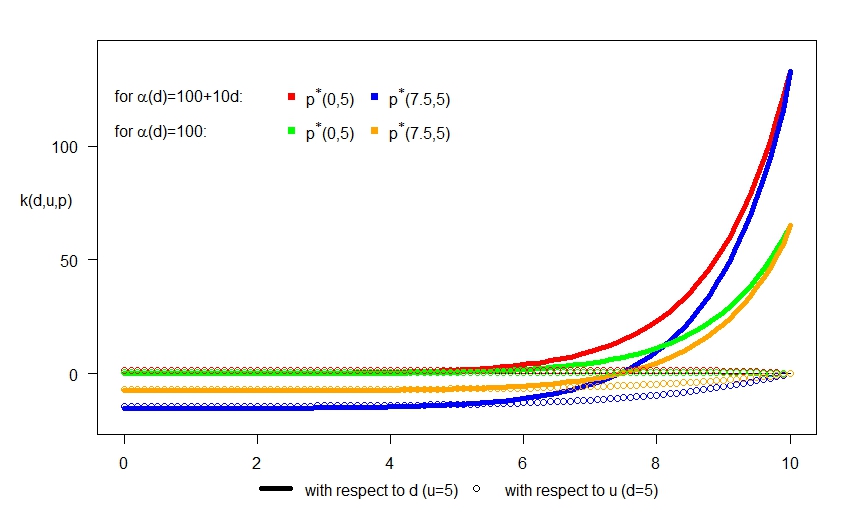}
\caption{\footnotesize{The value of the contract $k(d,u,p)$ with drawup contingency for Cram\'er-Lundberg model
depending on the starting positions of drawdown and drawup for linear reward function $\alpha(d)$ and various premium levels $p$. Parameters: $r=0.01,\ \hat{\mu}=0.05,\ \beta=0.1,\ \rho=2.5,\ a=10.$}}
\label{fig:k_cl}
\end{figure}

\begin{figure}[!ht]
\center{
\subfloat{\includegraphics[width=0.60\textwidth]{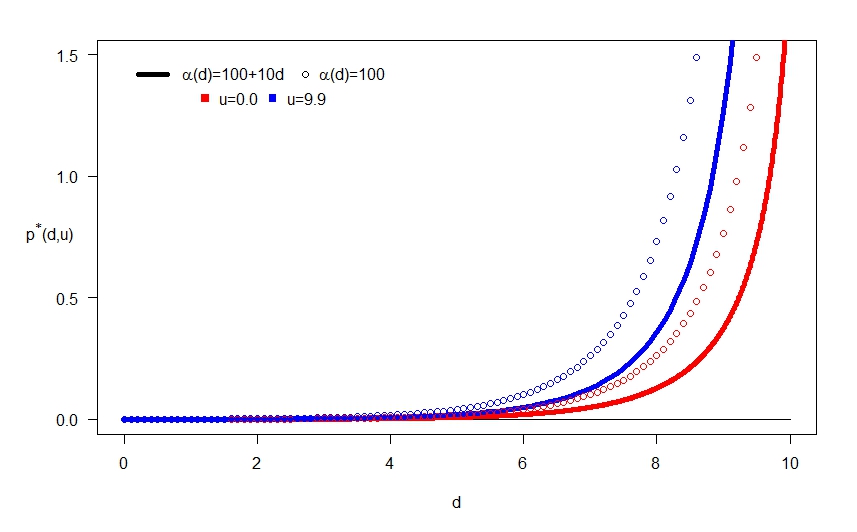}}
\quad
\subfloat{\includegraphics[width=0.6\textwidth]{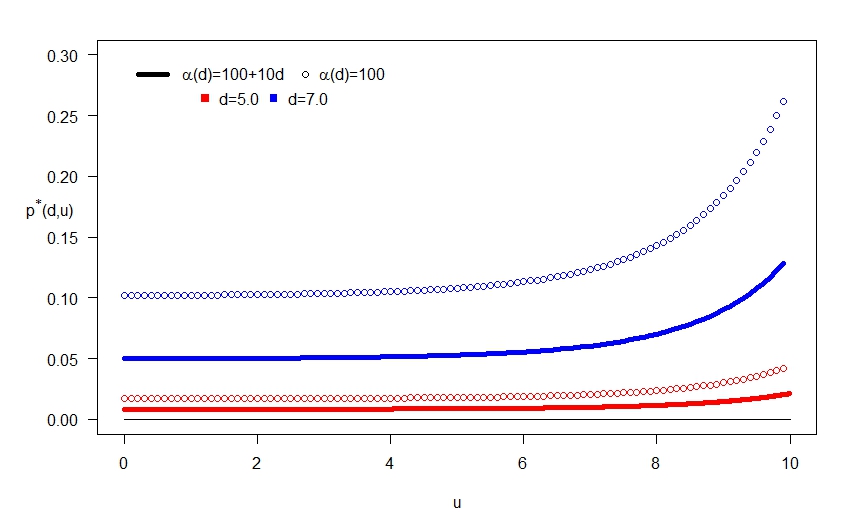}}
\caption{\footnotesize{The fair premium $p^*$ for the contract
$k(d,u,p)$ with drawup contingency for the Cram\'er-Lundberg model depending on the starting positions of drawdown (top) and drawup (bottom) for linear reward function $\alpha(d)$.
Parameters: $r=0.01,\ \hat{\mu}=0.05,\ \beta=0.1,\ \rho=2.5,\ a=10.$}}\label{fig:p*_k_cl}}
\end{figure}
\end{Ex}

\subsection{Cancellable feature}
The last contract that we consider in this paper allows investors to terminate it earlier and it has drawup contingency. In other words, we
add the cancellable feature to the previous contract $k(d,u,p)$ given in \eqref{kyzp}.
In this case, the protection buyer can terminate his or her position by paying fee $c(D_\tau)\geq 0$ at any time prior to drawdown epoch.
Note that the value of this fee may depend on the value of the drawdown at the moment of termination.
The value of this contract then equals,
\begin{align}
K(d,u,p):=&\sup\limits_{\tau\in\mathcal{T}}\mathbb{E}_{|d|u}\Bigg[-\int_0^{\tau^+_D (a)\wedge\tau^+_U(b)\wedge\tau}e^{-rt}p\diff t\nonumber\\
&+\alpha(D_{\tau^+_D (a)}) e^{-r\tau^+_D(a)}\mathbbm{1}_{\left(\tau^+_D(a)<\tau^+_U(b)\wedge\tau\right)}-c(D_\tau)e^{-r\tau}\mathbbm{1}_{\left(\tau<\tau^+_D(a)\wedge\tau^+_U(b)\right)};\ \tau<\infty\Bigg].\label{K}
\end{align}
To analyse the cancellable feature, we rewrite the contract value in a similar form as that given in Lemma \ref{war}.
That is, we can represent the drawup contingency contract with cancellable feature as the sum of two parts:
one without cancellable feature (it is then $k(d,u,p)$) and one that depends on a stopping time $\tau$.
\begin{Prop}\label{Kdecompose}
The cancellable drawup insurance value \eqref{K} admits the following decomposition:
\begin{align}\label{kdecompose}
K(d,u,p)=k(d,u,p)+H(d,u,p),
\end{align}
where
\begin{align}
&H(d,u,p):=\sup\limits_{\tau\in\mathcal{T}}h_{\tau}(d,u,p),\label{fh}\\
&h_{\tau}(d,u,p):=\mathbb{E}_{|d|u}\left[e^{-r\tau}\tilde{k}(D_{\tau},U_{\tau},p);\tau<\tau_D^+(a)\wedge\tau_U^+(b),\ \tau<\infty\right],\label{htau}\\
&\tilde{k}(d,u,p):=-k(d,u,p)-c(d)\label{ktilde}
\end{align}
for $k$ defined in \eqref{kyzp} and identified in (\ref{fdrawup}).
\end{Prop}
\begin{proof} Using $\mathbbm{1}_{(\tau_D^+(a)<\tau_U^+(b)\wedge\tau)}=\mathbbm{1}_{(\tau_D^+(a)<\tau_U^+(b))}-\mathbbm{1}_{(\tau<\tau_D^+(a)<\tau_U^+(b))}$ we obtain:
\begin{align}
&K(d,u,p)=\mathbb{E}_{|d|u}\left[-\int_0^{\tau_D^+ (a)\wedge\tau_U^+(b)}e^{-rt}p\diff t+\alpha(D_{\tau_D^+(a)}) e^{-r\tau_D^+(a)}\mathbbm{1}_{\left(\tau_D^+(a)<\tau_U^+(b)\right)}\right]\nonumber\\
&\quad+\sup\limits_{\tau\in\mathcal{T}}\mathbb{E}_{|d|u}\left[\int_{\tau\wedge\tau_D^+(a)\wedge\tau_U^+(b)}^{\tau_D^+(a)\wedge\tau_U^+(b)}e^{-rt}p\diff t-\alpha(D_{\tau_D^+(a)})e^{-r\tau_D^+(a)}\mathbbm{1}_{\tau<\tau_D^+(a)<\tau_U^+(b)}-c(D_{\tau})e^{-r\tau}\mathbbm{1}_{(\tau<\tau_D^+(a)\wedge\tau_U^+(b))};\ \tau<\infty\right].\nonumber
\end{align}
Note that $\mathbbm{1}_{(\tau<\tau_D^+(a)<\tau_U^+(b)}=\mathbbm{1}_{(\tau<\tau_D^+(a)\wedge\tau_U^+(b)}\mathbbm{1}_{(\tau_D^+(a)<\tau_U^+(b)}$. Thus, the result follows from the Strong Markov Property applied at the stopping time $\tau$:
\begin{align}
K(d,u,p)=&k(d,u,p)+\sup\limits_{\tau\in\mathcal{T}}\mathbb{E}_{|d|u}\Big[e^{-r\tau}\mathbb{E}_{|D_{\tau}|U_{\tau}}\left[\int_{0}^{\tau_D^+(a)\wedge\tau_U^+(b)}e^{-rt}p\diff t\right]\nonumber\\
&\qquad\qquad\qquad-\alpha(D_{\tau_D^+(a)})e^{-r\tau_D^+(a)}\mathbbm{1}_{(\tau_D^+(a)<\tau_U^+(b))}-c(D_{\tau})e^{-r\tau};\ \tau<\tau_D^+(a)\wedge\tau_U^+(b),\ \tau<\infty\Big].\nonumber
\end{align}
\end{proof}

To identify the value of the contract $K$ and, hence, to find the function $H$ defined in \eqref{fh}
we use, as before, the ``guess and verify'' approach.
The candidate for the optimal stopping strategy $\tau^*$ is the first passage time over level
$d-\theta^*$ for some $\theta^* \in (d+u-b,a)$ that will be specified later; that is,
\begin{equation}\label{tau*drugie}
\tau^*:=
\tau_{d-\theta^*}^+.
\end{equation}
\begin{Rem}\label{negativetheta}\rm
We allow here $\theta^*$ to be negative, which corresponds to the rise of the running supremum from $d$ to $d+|\theta^*|$.
Finally, note that, if $\theta^*>0$ then $\tau^*$ becomes $\tau_D^-(\theta^*)$, as for the
drawdown cancellable contract without drawup contingency.
For example, by using Verification Lemma \ref{war}, to simplify exposition, for $\theta^*<0$ we will later
treat $\tau_{d-\theta^*}^+$ as the first drawdown passage time $\tau_D^-(\theta^*)$.
\end{Rem}
At this point it should be noted that if $\tilde{k}(D_{\tau_{d-\theta}^+},U_{\tau_{d-\theta}^+},p)<0$ for all $\theta$, then
it is optimal for the investor to never terminate the contract and, hence, $\tau=\infty$ and $h_{\tau}(d,u,p)$=0.
To avoid this trivial case, from now we assume that there exists at least one point $(d_0,u_0)$ with $d_0\in[0,d)$ and $u_0>u$ for which $\tilde{k}(d_0,u_0,p)>0$.
This is equivalent to the following inequality:
\begin{align}
\frac{p}{r}-c(d_0)>\frac{p}{r}\nu(d_0,u_0)+\frac{p}{r}\lambda(d_0,u_0)+N(d_0,u_0)\geq 0.\label{war2}
\end{align}

We will now find function $h_{\tau}(d,u,p)$ defined in \eqref{htau}
for the postulated stopping rule
\begin{equation}\label{tauzwykle}\tau= \tau_{d-\theta}^+\qquad\text{for $\theta \in (d+u-b,d)$.}\end{equation}
Then, we will maximise $h_{\tau}(d,u,p)$ over $\theta$ to find $\theta^*$. In the last step, we
will utilise the Verification Lemma (checking all its conditions) to verify that the suggested
stopping rule is a true optimal.

We start from the first step, which is identifying the function $h_{\tau}(d,u,p)$
for $\tau$ given \eqref{tauzwykle}.
For $\theta \in (d+u-b,a)$ we denote:
\begin{align}
h_{\tau_{d-\theta}^+}(d,u,p)=h_{\tau_{d-\theta}^+}(d,u,p)\mathbbm{1}_{(d>\theta)}+h_{\tau_{d-\theta}^+}(d,u,p)\mathbbm{1}_{(d\leq\theta)}:=h_>(d,u,p,\theta)\mathbbm{1}_{(d>\theta)}+h_<(d,u,p,\theta)\mathbbm{1}_{(d\leq\theta)}.\label{h1}
\end{align}
Recall that the underlying process $X$ starts at 0.
Note  that for $\theta\geq d$, then $X_0>d-\theta$ and the investor should stop the contract immediately; that is,
\begin{align}
h_<(d,u,p,\theta)=\tilde{k}(d,u,p).\label{h=k}
\end{align}
Assume now that $\theta<d$. We will calculate function $h_>(d,u,p,\theta)$.
Given that the considered process $X_t$ has no strictly positive jumps and, hence, it crosses upward at all levels continuously,
from the definitions of the drawup and drawdown processes we have:
\begin{align}
U_{\tau^+_{d-\theta}}=&X_{\tau^+_{d-\theta}}-\underline{X}_{\tau^+_{d-\theta}}\wedge (-u)=d-\theta-\underline{X}_{\tau^+_{d-\theta}}\wedge (-u),\nonumber\\
D_{\tau^+_{d-\theta}}=&\overline{X}_{\tau^+_{d-\theta}}\vee d-X_{\tau^+_{d-\theta}}=(d-\theta)\vee d - (d-\theta)=\theta\vee 0.\nonumber
\end{align}
Moreover, on the event $\{\tau^+_{d-\theta}<\tau_U^+(b)\wedge\tau_D^+(a)\}$ we can observe that $\overline{U}_{\tau^+_{d-\theta}}<b$ and $\overline{D}_{\tau^+_{d-\theta}}<a$,
where $\overline{U}_t:=\sup_{s\leq t} U_s$ and $\overline{D}_t:=\sup_{s\leq t} D_s$ denote the running supremum of the processes $U_t$ and $D_t$, respectively.
In fact, from the definition of the drawup we have $\overline{U}_{\tau^+_{d-\theta}}=U_{\tau^+_{d-\theta}}$.
Thus, the following inequality holds true: $\underline{X}_{\tau^+_{d-\theta}}\wedge (-u)>d-\theta-b$.
When analysing the drawdown process, we consider two cases. First, let $\theta>0$, and then $\overline{D}_{\tau^+_{d-\theta}}=d-\underline{X}_{\tau^+_{d-\theta}}<a$. Second, if $\theta\leq 0$, then we have $\overline{D}_{\tau^+_{d-\theta}}\leq (d-\theta)-\underline{X}_{\tau^+_{d-\theta}}\leq U_{\tau^+_{d-\theta}}<b<a$.
These observations give an additional inequality:
$\underline{X}_{\tau^+_{d-\theta}}>(d-\theta)\vee d-a$.
We can now rewrite the function $h_>$  as follows:
\begin{align}
h_>(d,u,p,\theta)=&\E_{|d|u}\left[e^{-r\tau^+_{d-\theta}}\tilde{k}(\theta\vee 0,d-\theta-\underline{X}_{\tau^+_{d-\theta}}\wedge (-u),p);\ \underline{X}_{\tau^+_{d-\theta}}>(d-\theta)\vee d-a,\ \underline{X}_{\tau^+_{d-\theta}}\wedge (-u)>d-\theta-b\right]\nonumber\\
=&\tilde{k}(\theta\vee 0,d+u-\theta,p)\E_{|d|u}\left[e^{-r\tau^+_{d-\theta}};\ \underline{X}_{\tau^+_{d-\theta}}>(-u)\vee ((d-\theta)\vee d-a)\right]\mathbbm{1}_{(d+u-b<\theta)}\nonumber\\
&+\E_{|d|u}\left[e^{-r\tau^+_{d-\theta}}\tilde{k}(\theta\vee 0,d-\theta-\underline{X}_{\tau^+_{d-\theta}},p);\ -u\geq\underline{X}_{\tau^+_{d-\theta}}>(d-\theta-b)\vee ((d-\theta)\vee d-a)\right].\nonumber
\end{align}
The joint law for $(\tau_{d-\theta}^+,\underline{X}_{\tau^+_{d-\theta}})$ can be derived from two-sided formula given in \eqref{twosided1}:
\begin{align}
\E\left[e^{-r\tau^+_{d-\theta}};\ \underline{X}_{\tau^+_{d-\theta}}>-y\right]=\E\left[e^{-r\tau_{d-\theta}^+};\ \tau_{d-\theta}^+<\tau_{-y}^-\right]=\frac{W^{(r)}(y)}{W^{(r)}(d-\theta+y)}.\nonumber
\end{align}
The final form of the function $h_>(d,u,p,\theta)$ can  then be expressed as follows:
\begin{align}
h_>(d,u,p,\theta)=&\tilde{k}(\theta^+,(d-\theta+u)^+,p)\frac{W^{(r)}(u\wedge (a-(d-\theta)\vee d))}{W^{(r)}(d-\theta+u\wedge (a-(d-\theta)\vee d))}\mathbbm{1}_{(\theta>d+u-b)}\nonumber\\
&+\int_{u}^{(b+\theta-d)\vee (a-(d-\theta)\vee d)}\tilde{k}(\theta^+,(d-\theta+y)^+,p)\frac{\partial}{\partial y}\frac{W^{(r)}(y)}{W^{(r)}(d-\theta+y)}\diff y \mathbbm{1}_{((b+\theta-d)\vee (a-(d-\theta)\vee d)>u)}.\label{h>}
\end{align}

Because $h_<(d,u,p,\theta)$ does not depend on $\theta$, as previously for the drawdown contract,
we choose the optimal level $\theta^*$ such that it maximises the function $h_>(d,u,p,\theta)$; that is, we define:
\begin{align}\label{tetadu*}
\theta^*=\inf\left\{ \theta\in (d+u-b,d]:\ h_>(d,u,p,\varsigma)\leq h_>(d,u,p,\theta)\ \forall\varsigma\geq 0\ \right\}.
\end{align}
\begin{lem}
Assume that \eqref{war2} holds. Then, there exist $\theta^*$ defined in \eqref{tetadu*}.
\end{lem}
\begin{proof}
First note that $h_>(d,u,p,\cdot)$ is continuous on $[d+u-b,d]$ because $\tilde{k}$ and scale function $W^{(r)}$ are
continuous.  If $\theta\downarrow d+u-b$ then $D_{\tau^+_{d-\theta}}\downarrow (d+u-b)^+$ and $U_{\tau^+_{d-\theta}}\uparrow b$. Thus,
\begin{align}
h_>(d,u,p,\theta)\longrightarrow-c((d+u-b)^+)\frac{W^{(r)}(u\wedge((a-b+u)\vee d))}{W^{(r)}(b-u+u\wedge((a-b+u)\vee d))}\leq 0.\nonumber
\end{align}
On the one hand, as $\theta\uparrow d$ we get $D_{\tau^+_{d-\theta}}\uparrow d$ and then,
\begin{align}
h_>(d,u,p,\theta)\longrightarrow \tilde{k}(d,u,p).\nonumber
\end{align}
If $\tilde{k}(d,u,p)<0$ then by assumption \eqref{war2} there exists $\theta^*$.
On the other hand, when $\tilde{k}(d,u,p)\geq 0$  then either $\theta^*$ exists inside interval $(d+u-b,d)$ or the maximum is attained at boundary; that is, $\theta^*=d$. The second scenario corresponds to a case
where the contract is immediately stopped.
\end{proof}

From this proof, we can see that $h_>(d,u,p,\theta)\longrightarrow \tilde{k}(d,u,p)$.
Thus, the {\it continuous fit} for this problem is always satisfied.
\footnote{
Because of the complexity of the value function $h_>(t,\theta^*,u)$ we cannot prove the {\it smooth pasting} condition
using theoretical considerations. However, we believe that it always holds true.}

In the last step, using Verification Lemma \ref{war}, we will prove that the stopping rule \eqref{tau*drugie}
is indeed optimal.
\begin{theorem}\label{th2}
Let $c: [d+u-b,a)\rightarrow\mathbb{R}_+$ be a non-increasing bounded function in $C^2([d+u-b,a))$ and it additionally satisfies \eqref{war1}.
Assume that \eqref{war2} holds. Let $\theta^*$ defined by \eqref{tetadu*} satisfy the smooth fit condition:
$$\frac{\partial}{\partial d}h_>(d,u,p,\theta^*)|_{d\searrow\theta^*}=\frac{\partial}{\partial d}h_<(d,u,p,\theta^*)|_{d\nearrow\theta^*}$$
if $\sigma>0$.
If
\begin{align}\label{assum_additional2}
\int_{((\theta^*-d)^+,\infty)}\tilde{k}(d+z,(u-z)\vee 0,p)\Pi(-\diff z)\geq 0\textrm{ for all }d\in[0,\theta^*),
\end{align}
then $\tau^*=\tau_{d-\theta^*}^+$ given in \eqref{tau*drugie}
is the optimal stopping rule for the stopping problem \eqref{fh} and, hence, for the insurance contract \eqref{K}.
In this case the value $K(d,u,p)$ of the drawdown contract with
drawup contingency and cancellable feature equals $k(d,u,p)+h_{\tau_{d-\theta^*}^+}(d,u,p)$, as defined in \eqref{h1}
for $h_<(d,u,p,\theta^*)$, and identified in \eqref{h=k} and for $h_>(d,u,p,\theta^*)$ given in \eqref{h>}.
\end{theorem}

\begin{proof}
We again apply the Verification Lemma \ref{war} for the optimal stopping problem \eqref{fh}.
This time we take Markov process $\Upsilon_t=(D_t, U_t)$, $\mathbb{B}=\mathbb{R}_+\times\mathbb{R}_+$, $\tau_0=\tau_U^+(b)\wedge\tau_D^+(a)$ and
$V(\phi)=\tilde{k}(d,u,p)$ with $\phi=(d,u)$.
To prove the domination condition (i) of Verification Lemma \ref{war},
we have to show that $h_{\tau^+_{d-\theta^*}}(d,u,p)\geq \tilde{k}(d,u,p)$.
To do this, we take $\theta=d$ and use the definition of $\theta^*$ in \eqref{tetadu*}.
In particular, if $\theta^*<d$ then
\begin{align}
h_{\tau^+_{d-\theta^*}}(d,u,p)=h_>(d,u,p,\theta^*)\geq h_>(d,u,p,d)=\tilde{k}(d,u,p).\nonumber
\end{align}
Otherwise, if $\theta^*\geq d$, then:
\begin{align}
h_{\tau^+_{d-\theta^*}}(d,u,p)=h_<(d,u,p,\theta^*)=\tilde{k}(d,u,p).\nonumber
\end{align}

To prove the second condition of Verification Lemma \ref{war}, which states
that process
\begin{align}
e^{-r(t\wedge\tau_D^+(a)\wedge\tau_U^+(b))}h(D_{t\wedge\tau_D^+(a)\wedge\tau_U^+(b)},U_{t\wedge\tau_D^+(a)\wedge\tau_U^+(b)},p,\theta^*)\nonumber
\end{align}
is a supermartingale,
we use key Lemma \ref{supermtg}.
We take $p(t,d,u):=e^{-rt}h(d,u,p)$ for $d\in[0,\infty)$, $u\in[0,b)$  and $C=(\theta^*,\infty)$, $S=[d+u-b,\theta^*)$, and then $p_C(t,d,u)=e^{-rt}h_>(t,d,u,\theta^*)\mathbbm{1}_{(d\in C)}$ and $p_S(t,d,u)=e^{-rt}h_<(t,d,u,\theta^*)\mathbbm{1}_{(d\in S)}$.
Note that $\theta^*$ always satisfies the continuous fit and the smooth fit condition by the made assumption.
Also observe that function $h_{\tau_{d-\theta^*}^+}(d,u,p)$ defined in \eqref{htau} is bounded.
Indeed, this follows from
the inequalities $0\leq \nu(d,u)\leq 1$, $0\leq \lambda(d,u)\leq 1$,
$N(d,u)\leq \Xi(d)<\infty$ by \eqref{warXi} for any $d\geq 0, u\geq 0$.
Thus, from the Strong Markov Property, the process
\begin{align}
e^{-r(t \tau_D^-(\theta^*))}h(D_{t\wedge\tau_D^-(\theta^*)},U_{t\wedge\tau_D^-(\theta^*)},p,\theta^*)\nonumber
\end{align}
is a martingale (with the interpretation as mentioned in Remark \ref{negativetheta} for $\theta ^*\leq 0$).
Thus, the condition from Lemma \ref{supermtg} (ii) is satisfied and $I(p_C)$ is a martingale.

To prove that $I(p_S)$ is a supermartingale for region $S$,
we first show that $e^{-rt}\tilde{k}(D_t,U_t,p)$ is a supermartingale.
Indeed, from the definition of $\tilde{k}$ given in \eqref{ktilde} we have
\begin{equation}
\tilde{k}(d,u,p)=-\frac{p}{r}\nu(d,u)-\frac{p}{r}\lambda(d,u)-N(d,u)+\frac{p}{r}-c(d).\nonumber
\end{equation}
Now, note that
\begin{align}
e^{-r t}\nu(D_{t},U_{t}),\quad e^{-rt}\lambda(D_{t},U_{t}),\quad e^{-rt}N(D_{t},U_{t})\nonumber
\end{align}
are $\mathcal{F}_t$-martingales.
Then, the supermartingale condition of $e^{-rt}\tilde{k}$ holds if $\left(\frac{p}{r}-c(d)\right)$ is a supermartingale.
By \cite[Thm. 31.5]{sato} this holds when:
\begin{align}
&\mathcal{A}^{(D,U)}\left(\frac{p}{r}-c(d)\right)-r\left(\frac{p}{r}-c(d)\right)\leq 0,\nonumber
\end{align}
which is a consequence of assumption \eqref{war1}; see also,
\eqref{adu}.
Note that the condition \eqref{normalref2} also holds from assumption that function $c$ is non-increasing.
Indeed, observe that
\begin{align}
\frac{\partial}{\partial d}\left(\frac{p}{r}-c(d)\right)=-c^\prime (d)\geq 0\nonumber
\end{align}
and $\frac{\partial}{\partial u}\left(\frac{p}{r}-c(d)\right)=0$.
To prove the supermartingale property of $I(p_S)$ for $p_S(t,D_t,U_t)=e^{-rt}h_<(D_t,U_t,p)\mathbbm{1}_{(d\in S)}$ note that for $d\in S$ we have $p_S(t,d,u)=e^{-rt}\tilde{k}(d,u,p)$ from \eqref{h=k}. Thus, the generator $\mathcal{A}^{(D,U)}$ of $p_S$, as defined in \eqref{Dgenerator}, is equal:
\begin{align}
\mathcal{A}^{(D,U)}p_S(t,d,u)=e^{-rt}\mathcal{A}^{(D,U)}\tilde{k}(d,u,p)-\int_{((\theta^*-d)^+,\infty)}\tilde{k}(d+z,(u-z)\vee 0,p)\Pi(-\diff z)\geq 0.\nonumber
\end{align}
From the made assumption \eqref{assum_additional2}, the integral in the above formula is positive. Thus, for $d\in S$, we have
\begin{align}
-rp_S(t,d,u)+\mathcal{A}^{(D,U)}p_S(t,d,u)\leq -re^{-rt}\tilde{k}(d,u,p)+e^{-rt}\mathcal{A}^{(D,U)}\tilde{k}(d,u,p)\leq 0.\nonumber
\end{align}
Hence, the {\it supermartingale} condition \eqref{supmtg} from Lemma \ref{supermtg} (iii) for $I(p_S)$ is satisfied. Given that $p_S(t,d,u)=e^{-rt}\tilde{k}(d,u,p)$ for $d\in S$, the condition \eqref{normalref2} is the same for both functions because the derivative in \eqref{normalref2} is checked at $d+u-b\in S$. Thus, the condition \eqref{normalref2} is satisfied for $p_S$ because it is satisfied for $e^{-rt}\tilde{k}(d,u,p)$, as we have shown earlier. This completes the proof of the supermartingale property of $I(p_S)$ for $p_S(t,d,u)=e^{-rt}h_<(d,u,p,\theta^*)\mathbbm{1}_{(d\in S)}$.
\end{proof}

\addtocounter{Ex}{-2}
\begin{Ex}[continued]
We continue the analysis of linear Brownian motion for the drawdown contract with drawup contingency and cancellable feature.
We take $b<a$.
Note that we can calculate function $h_>$ from \eqref{h>}.
Then, we can find $\theta^*$ that maximises this function $h_{>}$.
We can numerically check that for chosen parameters this identified $\theta^*$ indeed satisfies the smooth fit condition.
Note also that the condition \eqref{assum_additional2} is satisfied because the Brownian motion has no jumps.
This means that $\tau^+_{d-\theta^*}$ is the optimal stopping rule.

\begin{figure}[!ht]
\center{
\subfloat{\includegraphics[width=0.4\textwidth]{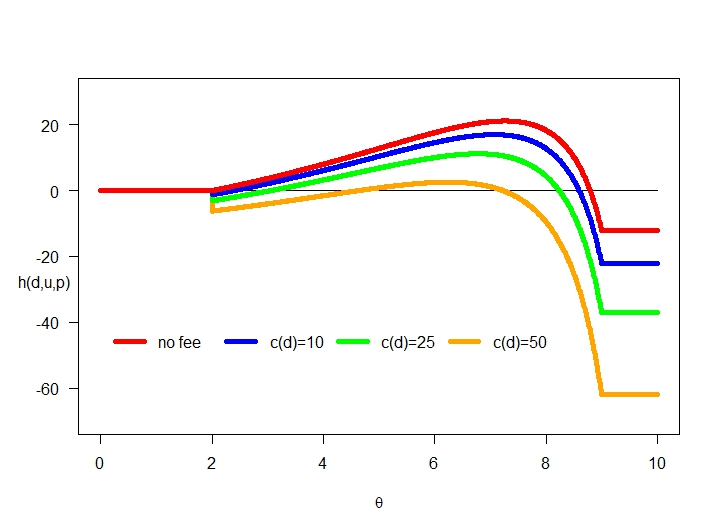}}
\subfloat{\includegraphics[width=0.4\textwidth]{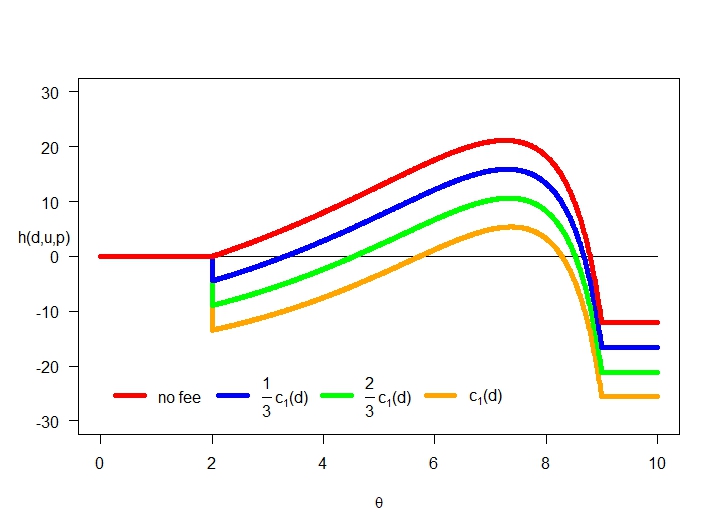}}\\
\subfloat{\includegraphics[width=0.4\textwidth]{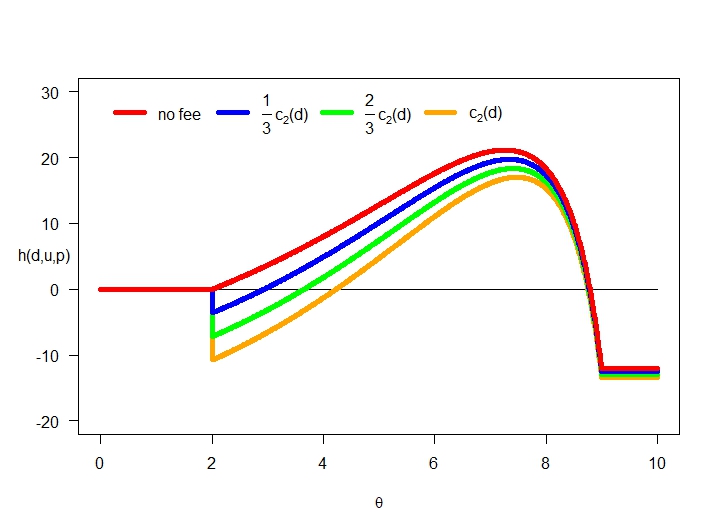}}
\caption{\footnotesize{The function $h$ for linear Brownian motion and various fee functions: constant (first), linear (second) and quadratic (third). Parameters: $r=0.01,\ \mu=0.04,\ \sigma=0.3,\ a=10, \alpha=100, p=1.35, d=9, u=1.$}}\label{fig:h_bm}}
\end{figure}
On Figure \ref{fig:h_bm} we depicted the function $h$ for constant, linear \eqref{c1_ex} and quadratic \eqref{c2_ex} fee functions.
Note that condition \eqref{war2} is satisfied in these cases.
On Figure \ref{fig:smooth_bm}, taking the fee function $c_2$ given in \eqref{c2_ex} (yellow graph on Figure \ref{fig:h_bm}),
we show that  $\theta^*\approx 1.8$, which was found before,
indeed satisfies the smooth fit condition in this case.
\begin{figure}[!ht]
\centering
\includegraphics[width=0.60\textwidth]{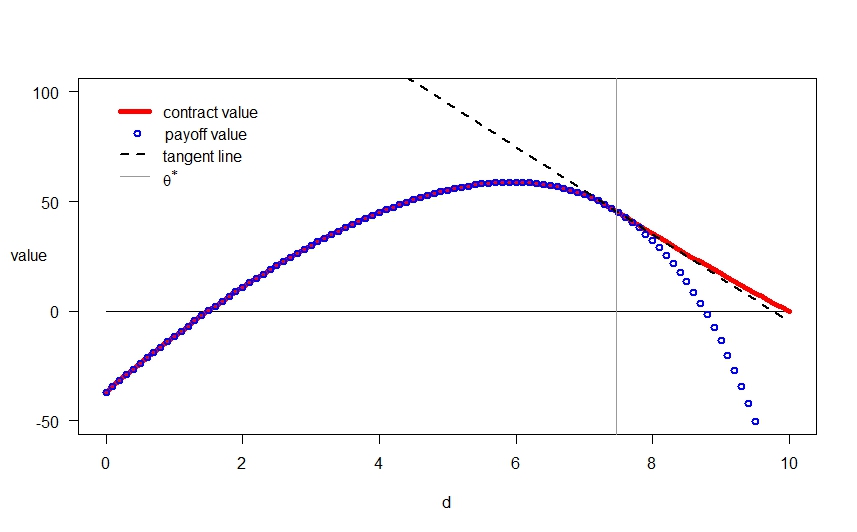}
\caption{\footnotesize{The smooth condition for linear Brownian motion and for quadratic fee function $c_2$. Parameters: $r=0.01,\ \mu=0.03,\ \sigma=0.4,\ a=10, b=8, \alpha=100, p=1.35, d=9, u=1.$}}
\label{fig:smooth_bm}
\end{figure}
\end{Ex}

\begin{Ex}[continued]
For the Cram\'er-Lundberg model, we consider the case when $a=b$.
Assume first that $d+u\geq a$, then the indicator in the last increment equals zero and we have:
\begin{align}
h_>(d,u,p,\theta)=&\mathbbm{1}_{(\theta>d+u-a)}\tilde{k}(\theta^+,(d-\theta+u)^+,p)\frac{W^{(r)}(u\wedge (a-(d-\theta)\vee d))}{W^{(r)}(d-\theta+u\wedge (a-(d-\theta)\vee d))}\nonumber
\end{align}
If $d+u<a$, then using integration by parts formula we get:
\begin{align}
&h_{>}(d,u,p,\theta)=\mathbbm{1}_{(\theta>d+u-a)}\Bigg(\tilde{k}(\theta^+,d-\theta+a-(d-\theta)\vee d,p)\frac{W^{(r)}(a-(d-\theta)\vee d)}{W^{(r)}(d-\theta+a-(d-\theta)\vee d)}\nonumber\\
&\quad-\frac{1}{r}\left(\Xi(0)\frac{W^{\prime (r)}(a)}{(W^{(r)}(a))^2}+\frac{p}{r}Z^{(r)}(a)\frac{W^{\prime (r)}(a)}{(W^{(r)}(a))^2}-\frac{p}{r}\frac{W^{\prime (r)}(a)}{(W^{(r)}(a))^2}-p\right)\left(Z^{(r)}(a-(d-\theta)\vee d)-Z^{(r)}(u)\right)\Bigg).\nonumber
\end{align}
Using these formula, we can find function $h_>$
for the linear fee function, which is defined as:
\begin{align}
c_3(d,c)=\frac{c-\frac{p}{r}}{a}d+\frac{p}{r}\mathbbm{1}_{(d<a)}\label{c3_ex}
\end{align}
for $c<\frac{p}{r}$. The graph of function $h_>$ is depicted on Figure \ref{fig:h_cl}. On the same Figure, we also show that condition \eqref{assum_additional2} is satisfied.
We check the condition \eqref{war2} numerically and find that it is satisfied in all considered cases.
\begin{figure}[!ht]

\centering
\center{
\subfloat{\includegraphics[width=0.5\textwidth]{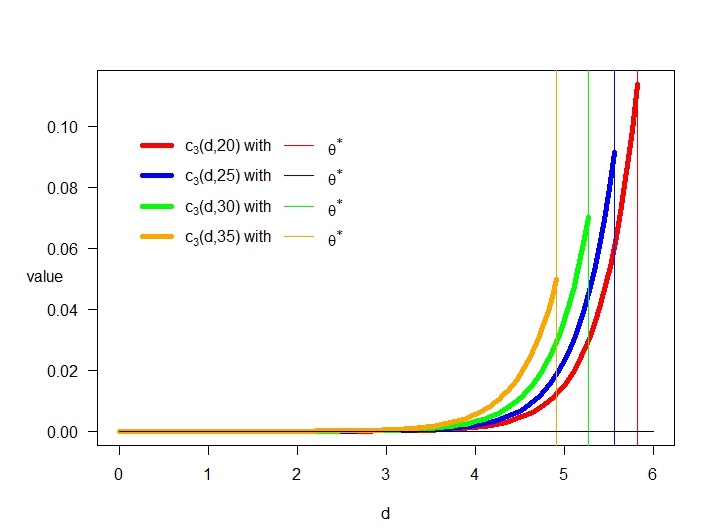}}
\subfloat{\includegraphics[width=0.50\textwidth]{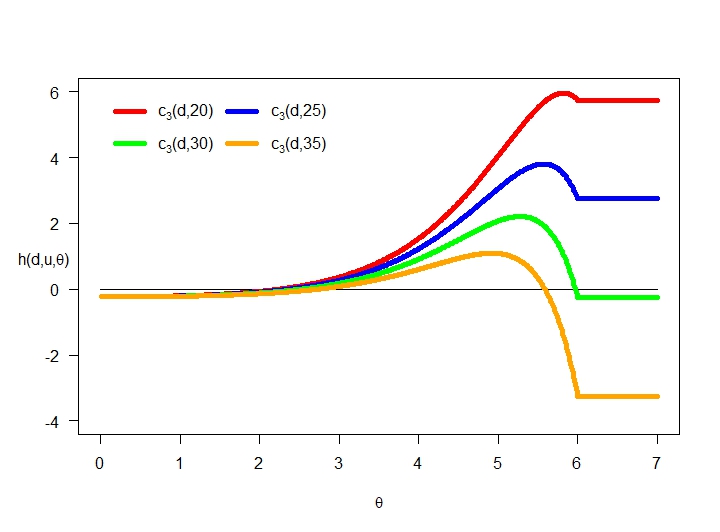}}}
\caption{\footnotesize{The condition \eqref{assum_additional2} (left) and the function $h$ (right) for Cram\'er-Lundberg model and linear fee function $c_3$. Parameters: $r=0.01,\ \hat{\mu}=0.04,\ \beta=0.1,\ \rho=2.5,\ a=10, \alpha(d)=100+20d, p=0.6.$}}
\label{fig:h_cl}
\end{figure}
On Figure \ref{fig:smooth_cl} we present the continuous fit pasting for yellow graph with $c_3(d,35)$, which
appears also on Figure \ref{fig:h_cl}.
\begin{figure}[!ht]
\centering
\includegraphics[width=0.60\textwidth]{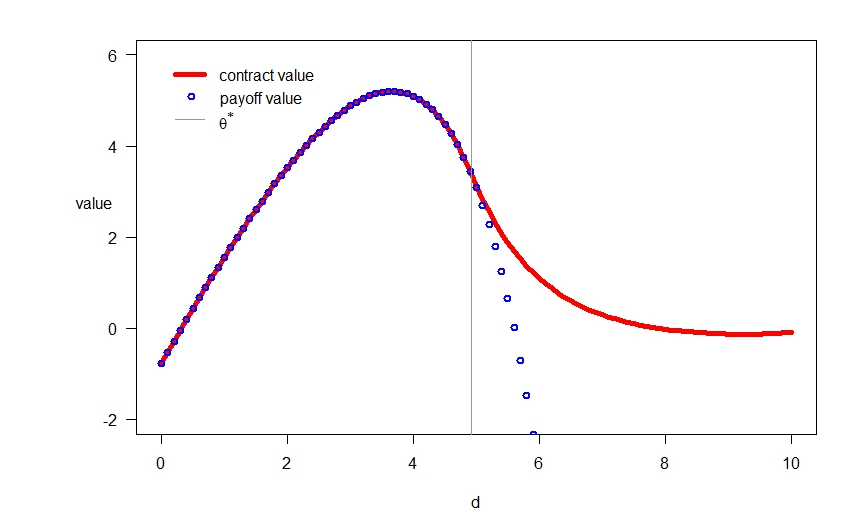}
\caption{\footnotesize{The continuous fit condition for the Cram\'er-Lundberg model with linear fee function $c_3(d,35)$. Parameters: $r=0.01,\ \hat{\mu}=0.04,\ \beta=0.1,\ \rho=2.5,\ a=10, \alpha(d)=100+20d, p=0.6.$}}
\label{fig:smooth_cl}
\end{figure}
\end{Ex}

\section{Appendix}
\begin{proof}[Proof of Proposition \ref{Ksi}]
The identity for $\xi$ was already proven in \cite{ZP&JT}. We focus on proving the identity for $\Xi$.
Note that $\tau_D^+(a)$ may happen before new supremum $\overline{X}_t>d$ is attained or after this event.
In a case where the underlying process $X$ does not cross level $d$, the drawdown epoch $\tau_D^+(a)$ just exceeds the level $d-a$ by the process $X$.
On the other hand, when $X$ crosses level $d$, then a new supremum is attained and the drawdown process starts from 0. We then split the function $\Xi$ into these two separate
scenarios:
\begin{align}
\Xi(d)=&\mathbb{E}_{|d}\left[e^{-r\tau_D^+(a)}\alpha(D_{\tau_D^+(a)});\tau_D^+(a)<\tau^+_d\right]+\mathbb{E}_{|d}\left[e^{-r\tau_D^+(a)}\alpha(D_{\tau_D^+(a)});\tau_d^+<\tau_D^+(a)\right]\nonumber\\
=&\mathbb{E}\left[e^{-r\tau_{d-a}^-}\alpha(d-X_{\tau_{d-a}^-});\tau_{d-a}^-<\tau_{d}^+\right]+\mathbb{E}\left[e^{-r\tau_d^+};\tau_d^+<\tau_{d-a}^-\right]\mathbb{E}\left[e^{-r\tau_D^+(a)}\alpha(D_{\tau_D^+(a)})\right]\nonumber\\
=&\mathbb{E}_{a-d}\left[e^{-r\tau_0^-}\alpha(a-X_{\tau_0^-});\tau_0^-<\tau_a^+\right]+\frac{W^{(r)}(a-d)}{W^{(r)}(a)}\mathbb{E}\left[e^{-r\tau_D^+(a)}\alpha(D_{\tau_D^+(a)})\right].\label{odnosnik}
\end{align}
We then use the two-sided formula given in \eqref{twosided1}.
We now identify the expectations appearing in the last line \eqref{odnosnik} of the above equation.

Let us rewrite the first expectation by considering the position of the process $X$ at time $\tau_0^-$:
\begin{align}
&\mathbb{E}_{a-d}\left[e^{-r\tau_0^-}\alpha(a-X_{\tau_0^-});\tau_0^-<\tau_a^+\right]=\alpha(a)\E_{a-d}\left[e^{-r\tau_0^-};\ \tau_0^-<\tau_a^+,\ X_{\tau^-_0}=0\right]\nonumber\\
&\qquad+\int_{(0,\infty)}\alpha(a+h)\E_{a-d}\left[e^{-r\tau_0^-};\ \tau_0^-<\tau_a^+,\ X_{\tau^-_0}\in -\diff h\right].\nonumber
\end{align}
The first increment of above identity refers to a case when the drawdown process $D_t$
creeps over level $a$ and the second increment refers to case when stopping time $\tau_D^+(a)$ is attained
by a jump of $X_t$ that puts the process $D_t$ strictly above level $a$.
The case of creeping was analysed in \cite[Cor. 3]{Pistorius}, producing:
\begin{align}
\E_{x}\left[e^{-r\tau_0^-};\ \tau_0^-<\tau_a^+,\ X_{\tau^-_0}=0\right]=\frac{u^r(x,0,a)}{u^r(0,0,a)},\nonumber
\end{align}
where
$u^r(x,\cdot,a)$ is a potential density of $X_t$ killed on exiting $[0,a]$ starting at $x\in[0,a]$.
By \cite[Th. 8.7]{KIntr} we have:
\begin{align}
u^r(x,y,a)=\frac{W^{(r)}(x)W^{(r)}(a-y)-W^{(r)}(x-y)W^{(r)}(a)}{W^{(r)}(a)}.\nonumber
\end{align}
We recall that the spectrally negative process $X_t$ creeps across $0$ if and only if $X_t$ has a non-zero Gaussian coefficient $\sigma^2$ (see \cite[Cor. 2]{Pistorius}).
Thus,
\begin{align}
\E_{a-d}\left[e^{-r\tau_0^-};\ \tau_0^-<\tau_a^+,\ X_{\tau^-_0}=0\right]
&=\lim_{{\epsilon\downarrow 0}}\frac{u^r(a-d+\epsilon ,\epsilon,a+\epsilon )}{u^r(\epsilon,\epsilon,a+\epsilon )}\nonumber\\
&=\lim_{{\epsilon\downarrow 0}}\frac{W^{(r)}(a-d+\epsilon )W^{(r)}(a)-W^{(r)}(a-d)W^{(r)}(a+\epsilon )}{W^{(r)}(\epsilon )W^{(r)}(a)-W^{(r)}(0)W^{(r)}(a+\epsilon )}\nonumber\\
&=\frac{W'^{(r)}(a-d)W^{(r)}(a)-W^{(r)}(a-d)W'^{(r)}(a)}{W'^{(r)}(0)W^{(r)}(a)-W^{(r)}(0)W'^{(r)}(a)}\nonumber\\
&=\frac{\sigma^2}{2}\left[W'^{(r)}(a-d)-\frac{W^{(r)}(a-d)W'^{(r)}(a)}{W^{(r)}(a)}\right].\nonumber
\end{align}
The case concerning jump can be solved by considering the joint law of $\tau_X^-(0)$ and $X_{\tau_X^-(0)}$, which is given in \cite[Th. 5.5]{KyprianouGerberShiu}
(note the result there, although presented only for the drift minus compound Poisson process, holds for general spectrally negative L\'evy process).
Then,
\begin{align}
\mathbb{E}_{x}\left[e^{-r\tau^-_0};X_{\tau^-_0}\in -\diff h,\ X_{\tau^-_0-}\in \diff z,\ \tau^-_0<\tau^+_a\right]=u^r(x,z,a)\Pi(-z-\diff h)\diff z,\nonumber
\end{align}
for $h>0$ and $z\in (0,a)$. Because we only need to know the position of $X_{\tau^-_0}$ ,we have to integrate above equation with respect to $z$.
In summary,
\begin{align}
&\mathbb{E}_{a-d}\left[e^{-r\tau_0^-}\alpha(a-X_{\tau_0^-});\tau_0^-<\tau_a^+\right]=\alpha (a)\frac{\sigma^2}{2}\left[W^{\prime (r)}(a-d)-W^{(r)}(a-d)\frac{W^{\prime (r)}(a)}{W^{(r)}(a)}\right]\nonumber\\
&\qquad +\int_0^a\int_{(0,\infty)}\alpha(a+h)\left[\frac{W^{(r)}(a-d)W^{(r)}(a-z)}{W^{(r)}(a)}-W^{(r)}(a-d-z)\right]\Pi(-z-\diff h)\diff z.\nonumber
\end{align}

To find \eqref{odnosnik}, in the last part of the proof we find the formula for $\mathbb{E}\left[e^{-r\tau_D^+(a)}\alpha(D_{\tau_D^+(a)})\right]$. Note that
\begin{align}
\mathbb{E}\left[e^{-r\tau_D^+(a)}\alpha(D_{\tau_D^+(a)})\right]=\alpha(a)\E\left[e^{-r\tau_D^+(a)};\ D_{\tau_D^+(a)}=a\right]+\int_{(0,\infty)}\alpha(a+h)\E\left[e^{-r\tau_D^+(a)};\ D_{\tau_D^+(a)}\in a+\diff h\right].\nonumber
\end{align}
The first equation describes a case where drawdown process creeps over $a$ and the second equation refers to a case where it strictly exceeds level $a$ by jump.
These two expectations were calculated in \cite{Mijatovic&Pistorius}, as follows:
\begin{align}
\mathbb{E}_{}\Big[e^{-r\tau_D^+(a)};D_{\tau_D^+(a)}\in a+\diff h\Big]=\int_0^a \left[W'^{(r)}(a-z)\frac{W^{(r)}(a)}{W'^{(r)}(a)}-W^{(r)}(a-z)\right]\Pi(-z-\diff h)\diff z\nonumber
\end{align}
and
\begin{align}
\mathbb{E}_{}\Big[e^{-r\tau_D^+(a)};D_{\tau_D^+(a)}=a\Big]=\frac{\sigma^2}{2}\left[W'^{(r)}(a)-W''^{(r)}(a)\frac{W^{(r)}(a)}{W'^{(r)}(a)}\right].\nonumber
\end{align}
Putting all of the increments together completes the proof.
\end{proof}

\begin{proof}[Proof of Lemma \ref{supermtg}]
Assume at the beginning that $\sigma >0$.
Using the same arguments as in the proof of Eisenbaum and Kyprianou \cite[Thm. 3]{Eisenbaum&Kyprianou},
and fact that by our assumptions the transient density of $D_t$ exists and hence
\[\mathbb{P}(D_t=\theta^*)=0 \qquad \text{for all $t>0$}\]
(see also \cite[Eq. (2.26)-(2.30)]{Peskir}),
we can extend It\^{o} formula for function $p(t,D_t,U_t)$ into the following change of variables formula:
\begin{align}
p(t,D_t, U_t)=&I(p_C)(t,D_t, U_t)+I(p_S)(t,D_t, U_t)+\frac{1}{2}\int_0^t \left(\frac{\partial p_S}{\partial d}-\frac{\partial p_C}{\partial d}\right)(s,d,u)|_{(s,d,u)=(s,D_{s-}, U_{s-})}\diff L^\theta_s,\label{pierwsztozsamosc}
\end{align}
where $L_t^\theta$ is a local time of $D_t$ at $\theta$, which can be defined formally as was done for the process $X_t$ in
\cite[Thm. 3]{Eisenbaum&Kyprianou}.
However, in this construction we use one crucial
observation. Note that the local time in \cite{Eisenbaum&Kyprianou} is defined along some continuous curve $b(t)$.
The local time $L_t^\theta$ of $D_t$ at point $\theta$ is the same as the local time of $X_t$ at $b(t)=\overline{X}_t-\theta$,
which is continuous because process $\overline{X}_t$ is continuous.
Moreover, the process $X_t$ in \cite{Eisenbaum&Kyprianou} lives in whole real line but the drawdown process $D_t$ lives on non-negative half-line $[0,\infty)$.
Therefore we have in $I(p_C)$ and in $I(p_S)$ additional integrals with respect to continuous parts of supremum and infimum processes.

Now, the \textit{smooth fit} reduces the change variables formula into the following identity:
\begin{align}
p(t,D_t, U_t)=I(p_C)(t,D_t, U_t)+I(p_S)(t,D_t, U_t).\nonumber
\end{align}
Thus, $p(t,D_t, U_t)$ is a supermartingale if $I(p_C)(t,D_t, U_t)$ is a martingale and $I(p_S)(t,D_t, U_t)$ is a supermartingale.
This completes the proof of the first part (i).

To prove the second (ii) and third (iii) parts, note that from the Dynkin's formula for $\hat{p}=p_C$ or $\hat{p}=p_S$ we have:
\begin{eqnarray}
I(\hat{p})(t,D_t, U_t)&=&I(\hat{p})(0,D_0, U_0)\nonumber\\
&&+\int_0^t\frac{\partial}{\partial d}\hat{p}(s,d,u)|_{(s,d,u)=(s,0, U_{s-})}\diff \overline{X}^c_s
-\int_0^t\frac{\partial}{\partial u}\hat{p}(s,d,u)|_{(s,d,u)=(s,D_{s-}, 0)}\diff \underline{X}^c_s\nonumber\\&&
+\int_0^t\left(\frac{\partial}{\partial s}\hat{p}(s,d,u)|_{(s,d,u)=(s,D_s, U_{s})}+\mathcal{A}^{(D,U)}\hat{p}(s,D_s, U_s)\right)\diff s + \mathcal{M}_t,\label{ostatnialinijka}
\end{eqnarray}
where $\mathcal{M}_t$ is a martingale part and $\mathcal{A}^{(D,U)}$ is the full generator of the Markov process $(D_t,U_t)$ defined in \eqref{Dgenerator}-\eqref{generator_bounded}.
Other explanation comes from identifying the drift as a compensator of $X^c$, from the equality $\diff [X]^c_s=\sigma \diff t$ and compensation formula applied to the
jump part (see e.g. \cite[Thm. 4.4, p. 95]{KIntr}).

Moreover, to prove (ii) part, observe that
\begin{equation}\label{nowe}
I(p_C)(t\wedge \tau_D^-(\theta^*)-,D_{t\wedge \tau_D^-(\theta^*)-}, U_{t\wedge \tau_D^-(\theta^*)-})=
p_C(t\wedge \tau_D^-(\theta^*)-,D_{t\wedge \tau_D^-(\theta^*)-}, U_{t\wedge \tau_D^-(\theta^*)-}),
\end{equation}
which follows from It\^{o} formula applied to $p_C$ in the open set $C$.
Note also that process $(t,D_t, U_t)$ goes from set $C$ to set $S$ in continuous way since $X_t$ is spectrally negative.
If $p_C(t\wedge \tau_D^-(\theta^*),D_{t\wedge \tau_D^-(\theta^*)}, U_{t\wedge \tau_D^-(\theta^*)})=p_C(t\wedge \tau_D^-(\theta^*)-,D_{t\wedge \tau_D^-(\theta^*)-}, U_{t\wedge \tau_D^-(\theta^*)-})$ is a martingale, then from \eqref{ostatnialinijka} applied to $\hat{p}=p_C$ and from \eqref{nowe} by taking small $t>0$ we can conclude that
\textit{martingale condition} \eqref{mtg} and \textit{normal reflection condition} \eqref{normalref} hold true. Thus again by \eqref{ostatnialinijka} the process
$I(p_C)(t,D_t, U_t)$ is a martingale for all $t\geq 0$.

Similarly, from \eqref{ostatnialinijka} we can conclude that the \textit{supermartingale condition} \eqref{supmtg} and \textit{normal reflection condition} \eqref{normalref2}
give the supermartingale property of $I(p_S)(t,D_t,U_t)$ (we also use the observation that $\underline{X}^c_s$ is a decreasing process).

All of these arguments remain true and are almost unchanged for the L\'evy processes \eqref{CL} of bounded variation (hence, for $\sigma=0$).
The only change that has to be made is changing \eqref{pierwsztozsamosc} into
\begin{align}
p(t,D_t, U_t)=&I(p_C)(t,D_t, U_t)+I(p_S)(t,D_t, U_t)+\frac{1}{2}\int_0^t \left(p_S-p_C\right)(s,D_{s-}, U_{s-})\diff L^\theta_s;\nonumber
\end{align}
see \cite{Kyprianou&Surya} for details.
In the next step, the continuous pasting condition should be applied and the rest of the proof is the same as before.
This completes the proof.
\end{proof}

\begin{proof}[Proof of Proposition \ref{lambdanu}]
The identities for $\lambda(\cdot,\cdot)$ and $\nu(\cdot,\cdot)$ were proven in \cite{ZP&JT} for both cases $b<a$ and $b=a$. We focus on identifying $N(\cdot,\cdot)$. Note that $\mathbbm{1}_{(\tau_D^+(a)<\tau_U^+(b))}=1-\mathbbm{1}_{(\tau_U^+(b)<\tau_D^+(a))}$ because $\tau_D^+(a)$ and $\tau_U^+(b)$ cannot happen at the same time. Thus, for $b\leq a$, we have
\begin{align}
&\mathbb{E}_{|d|u}\left[e^{-r\tau_D^+(a)}\alpha(D_{\tau_D^+(a)});\ \tau_D^+(a)<\tau_U^+(b)\right]\nonumber\\
&\quad=\mathbb{E}_{|d|u}\left[e^{-r\tau_D^+(a)}\alpha(D_{\tau_D^+(a)})\right]-\mathbb{E}_{|d|u}\left[e^{-r\tau_D^+(a)}\alpha(D_{\tau_D^+(a)});\ \tau_U^+(b)<\tau_D^+(a)\right]\nonumber\\
&\quad=\mathbb{E}_{|d|u}\left[e^{-r\tau_D^+(a)}\alpha(D_{\tau_D^+(a)})\right]-\mathbb{E}_{|d|u}\left[e^{-r\tau_U^+(b)}\mathbb{E}_{|D_{\tau_U^+(b)}|b}\left[e^{-r\tau_D^+(a)}\alpha(D_{\tau_D^+(a)})\right];\ \tau_U^+(b)<\tau_D^+(a)\right],\nonumber
\end{align}
which follows from the Strong Markov Property.
We will analyse the cases $b<a$ and $b=a$ separately.
Assume first that $b<a$.
We can extend the equivalent representation of the event $\left\{\tau_U^+(b)<\tau_D^+(a),\ D_0=y,\ U_0=z\right\}$
in terms of running supremum and infimum of underlying process $X_t$ given in \cite{ZP&JT} by adding the position of the drawdown process
$D_t$ at the stopping moment $\tau_U^+(b)$ as follows:
\begin{align}
&\left\{\tau_U^+(b)<\tau_D^+(a),\ D_{\tau_U^+(b)}\in\diff s,\ D_0=d,\ U_0=u\right\}=\left\{\tau^+_{b-u}<\tau^-_{(d-a)\vee(-u)},\ D_{\tau^+_{b-u}}=(d+u-b)^+ \right\}\nonumber\\
&\qquad\cup\left\{\overline{X}_{\tau_U^+(b)}\vee d-\underline{X}_{\tau_U^+(b)}\in\diff s+b,\ \underline{X}_{\tau_U^+(b)}\leq -u,\ \diff s\in((d+u-b)^+,a-b)\right\}.\nonumber
\end{align}
This is a purely geometric and pathwise observation.
Using this identity, we can derive the following equality:
\begin{align}
&\mathbb{E}_{|d|u}\left[e^{-r\tau_U^+(b)}\mathbb{E}_{|D_{\tau_U^+(b)}|b}\left[e^{-r\tau_D^+(a)}\alpha(D_{\tau_D^+(a)})\right];\ \tau_U^+(b)<\tau_D^+(a)\right]\nonumber\\
&\quad=\int_{\left[(d+u-b)^+,a-b\right)}\mathbb{E}_{|d|u}\left[e^{-r\tau_U^+(b)}\mathbb{E}_{|s|b}\left[e^{-r\tau_D^+(a)}\alpha(D_{\tau_D^+(a)})\right];\ \tau_U^+(b)<\tau_D^+(a),\ D_{\tau_U^+(b)}\in\diff s\right]\nonumber\\
&\quad = \mathbb{E}\left[e^{-r\tau^+_{b-u}};\tau^+_{b-u}<\tau^-_{(d-a)\vee(-u)}\right]\mathbb{E}_{|(d+u-b)^+|b}\left[e^{-r\tau_D^+(a)}\alpha(D_{\tau_D^+(a)})\right]\nonumber\\
&\qquad +\int_{((d+u-b)^+,a-b)}\mathbb{E}_{|s|b}\left[e^{-r\tau_D^+(a)}\alpha(D_{\tau_D^+(a)})\right]\mathbb{E}_{|d|u}\left[e^{-r\tau_U^+(b)}; \overline{X}_{\tau_U^+(b)}\vee d-\underline{X}_{\tau_U^+(b)}\in\diff s+b,\ \underline{X}_{\tau_U^+(b)}\leq -u\right].\nonumber
\end{align}
Using definition of $\Xi(\cdot)$ given in \eqref{Xi} we obtain the result for $b<a$.

For $b=a$, note that when $a\leq d+u$ then $\{\tau_U^+(a)<\tau_D^+(a)\} = \{\tau^+_{a-u}<\tau^-_{d-a}\}$. Moreover, since $a-u\leq d$ and the process $X_t$ has no positive jumps, then $\overline{X}_{\tau^+_U(a)}\vee d=\overline{X}_0\vee d=d$ and $X_{\tau_U^+(a)}=a-u$. This gives that in this case $D_{\tau_U^+(a)}=d+u-a$. On the other hand, when $a>u+d$ then $X_{\tau_U^+(a)}=\overline{X}_{\tau_U^+(a)}\in (d,a-u]$ (see \cite{ZP&JT} for details). Therefore, we get that $D_{\tau_U^+(a)}=0$.
We have just proven that $D_{\tau_U^+(a)}=(d+u-a)^+$. Thus,
\begin{align}
&\mathbb{E}_{|d|u}\left[e^{-r\tau_U^+(b)}\mathbb{E}_{|D_{\tau_U^+(b)}|b}\left[e^{-r\tau_D^+(a)}\alpha(D_{\tau_D^+(a)})\right];\ \tau_U^+(b)<\tau_D^+(a)\right]\nonumber\\
&\quad = \E_{|(d+u-a)^+|a}\left[e^{-r\tau_D^+(a)}\alpha(D_{\tau_D^+(a)})\right]\E_{|d|u}\left[e^{-r\tau_U^+(a)};\ \tau_U^+(a)<\tau_D^+(a)\right]\nonumber\\
&\quad = \Xi((d+u-a)^+)\lambda(d,u).\nonumber
\end{align}
This completes the proof.
\end{proof}


\end{document}